\RequirePackage{tikz}
\documentclass{sn-jnl}
\usepackage{amsmath}
\usepackage{stmaryrd}
\usepackage{tikz}
\usepackage{nicefrac}
\usepackage{graphicx}
\usepackage{xspace}						
\usepackage{slashbox}
\usepackage{xypic}
\usepackage{color}

\jyear{2021}%

\theoremstyle{thmstyleone}%
\newtheorem{theorem}{Theorem}
\newtheorem{lemma}[theorem]{Lemma}
\newtheorem{corollary}[theorem]{Corollary}

\theoremstyle{thmstyletwo}%

\theoremstyle{thmstylethree}%

\newcommand{\NP}{{\sf NP}}

\newcommand{\imgw}[2]{
	\begin{center}
		\includegraphics[width=#2\textwidth]{#1}
	\end{center}
}

\newcommand\circlenum[1]{\tikz[baseline=(char.base)]{
		\node[shape=circle,draw,inner sep=2pt] (char) {#1};}}

\begin{document}

\title{The Complexity of $L(p,q)$-Edge-Labelling} 
 

\author[1]{Ga\'etan Berthe}
\author*[2]{Barnaby Martin}\email{barnabymartin@gmail.com}
\author[2]{Dani\"el Paulusma}
\author[2]{Siani Smith}
%

\affil[1]{ENS de Lyon, Lyon, France}
\affil[2]{Department of Computer Science, Durham University, U.K.}




\abstract{
The {\sc $L(p,q)$-Edge-Labelling} problem is the edge variant of the well-known {\sc $L(p,q)$-Labelling} problem. It is equivalent to the {\sc $L(p,q)$-Labelling} problem itself if we restrict the input of the latter problem to line graphs.
So far, the complexity of  {\sc $L(p,q)$-Edge-Labelling} was only partially classified in the literature.
We complete this study for all $p,q\geq 0$ by showing that whenever $(p,q)\neq (0,0)$, the {\sc $L(p,q)$-Edge-Labelling} problem is \NP-complete. We do this by proving that for all $p,q\geq 0$ except $p=q=0$, there is an integer~$k$ so that {\sc $L(p,q)$-Edge-$k$-Labelling} is \NP-complete.
}




\maketitle

\section{Introduction}\label{s-intro}

This paper studies a problem that falls under the distance-constrained labelling framework.
Given any fixed nonnegative integer values $p$ and $q$,  an {\it $L(p,q)$-$k$-labelling} is an assignment of {\it labels} from $\{0,\ldots,k-1\}$ to the vertices of a graph such that adjacent
vertices receive labels that differ by at least $p$, and vertices connected by a path of  length~$2$ receive labels that differ by at least $q$~\cite{Ca11}. Some authors instead define the latter condition as being vertices at distance~$2$ receive labels which differ by at least $q$ (e.g. \cite{FKK01}). These definitions are the same so long as $p\geq q$ and much of the literature considers only this case (e.g. \cite{JKM09}). If $q>p$, the definitions diverge. For example, in an $L(1,2)$-labelling, the vertices of a triangle $K_3$ \textcolor{black}{can take} labels $\{0,1,2\}$ in the second definition but \textcolor{black}{need} $\{0,2,4\}$ in the first. We use the {\it first} definition, in line with \cite{Ca11}.
The decision problem of testing if for a given integer $k$, a given graph $G$ admits an $L(p,q)$-$k$-labelling is known as {\sc $L(p,q)$-Labelling}. If $k$ is {\it fixed}, that is, not part of the input, we denote the problem as {\sc $L(p,q)$-$k$-Labelling}.

The {\sc $L(p,q)$-labelling} problem
has been heavily studied, both from the combinatorial and computational complexity perspectives. For a starting point, we refer the reader to the 
comprehensive survey of Calamoneri~\cite{Ca11}.\footnote{See \texttt{http://wwwusers.di.uniroma1.it/\textasciitilde calamo/survey.html} for later results.} 
The {\sc $L(1,0)$-Labelling} is the traditional {\sc Graph Colouring} problem (COL), whereas {\sc $L(1,1)$-Labelling} is known as {\sc (Proper) Injective Colouring}~\cite{BJMPS20,BJMPS21,HKSS02} and {\sc Distance $2$ Colouring}~\cite{LR92,M83}. The latter problem is studied explicitly in many papers (see~\cite{Ca11}), just as is {\sc $L(2,1)$-Labelling}~\cite{GY92,JKM09,KM18} (see also~\cite{Ca11}). 
The {\sc $L(p,q)$-labelling} problem is also studied for special graph classes, see in particular~\cite{FGK08} for a complexity dichotomy for trees. 
Janczewski et al.~\cite{JKM09} proved that if $p>q$, then \textsc{$L(p,q)$-Labelling} is \NP-complete for planar bipartite graphs.

We consider the edge version of the problem. The {\it distance} between two edges $e_1$ and $e_2$ is the length of a shortest path that has $e_1$ as its first edge and $e_2$ as its last edge minus~$1$ (we say that $e_1$ and $e_2$ are {\it adjacent} if they share an end-vertex or equivalently, are of distance~$1$ from each other).
The {\sc $L(p,q)$-Edge-Labelling} problem considers an assignment of the labels to the edges instead of the vertices, and now the corresponding distance constraints are placed instead on the edges. 

In \cite{KM18}, the complexity of  {\sc $L(2,1)$-Edge-$k$-Labelling} is classified. It is in P for $k<6$ and is \NP-complete for $k\geq 6$. In \cite{Ma02}, the complexity of  {\sc $L(1,1)$-Edge-$k$-Labelling} is classified. It is in P for $k<4$ and is \NP-complete for $k\geq 4$. In this paper we complete the classification of the complexity of  {\sc $L(p,q)$-Edge-$k$-Labelling}  in the sense that, for all $p,q \geq 0$ except $p=q=0$, we exhibit $k$ so we can show {\sc $L(p,q)$-Edge-$k$-Labelling} is \NP-complete. That is, we do not exhibit the border for $k$ where the problem transitions from P to \NP-complete (indeed, we do not even prove the existence of such a border). The authors of \cite{KM18} were looking for a more general result, similar to ours, but found the case $(p,q)=(2,1)$ laborious enough to fill one paper \cite{Ma20}. In fact, their proof settles for us all cases where $p\geq 2q$. We now give our main result. 
 
\begin{theorem}
For all $p,q\geq 0$ except if $p=q=0$, there exists an integer~$k$ so that {\sc $L(p,q)$-Edge-$k$-Labelling} is \NP-complete.\label{thm:hauptsatz}
\end{theorem}

\begin{center}
\begin{table}[t]
\begin{tabular}{|l|l|l|l|}
\hline
Regime & Reduction from & Place in article & $k$ at least \\
\hline
$p=0$ and $q>0$ & $3$-COL & Section~\ref{sec:0-1} & $3q$ \\
$2\textcolor{black}{<} \nicefrac{q}{p}$ & NAE-$3$-SAT &  Section~\ref{a-2} & $(n-1)p+q+1$ \\
$1<\nicefrac{q}{p} \leq 2$ & NAE-$3$-SAT &  Section~\ref{sec:1-to-2}	& $5p+1$ \\
$\nicefrac{q}{p}= 1$ & $3$-COL & \cite{Ma02} & $4p$\\
$\nicefrac{2}{3}<\nicefrac{q}{p} \leq 1$ & $3$-COL &  Section~\ref{sec:3-col} & $3p+q+1$ \\
$\nicefrac{q}{p}=\nicefrac{2}{3}$ & $1$-in-$3$-SAT &  Section~\ref{sec:1-in-3} & $4p$ \\
$\nicefrac{1}{2}<\nicefrac{q}{p} < \nicefrac{2}{3}$ & $2$-in-$4$-SAT & Section~\ref{sec:2-in-4} & $p+4q+1$ \\
$0<\nicefrac{q}{p} \leq \nicefrac{1}{2}$ & NAE-$3$-SAT & Section~\ref{a-last} \cite{KM18} & $3p+1$ \\
$p>0$ and $q=0$ & $3$-COL & Section~\ref{s-pre} & $3p$ \\
\hline
\end{tabular}
\vspace*{0.1cm}
\caption{Table of results.
The fourth row follows from~\cite{Ma02} (which proves the case $p=q=1$) and applying Lemma~\ref{lem:gcd}. The eighth row is obtained from a straightforward generalization of the result in~\cite{KM18} for the case where $p=2$ and $q=1$. The fourth column gives the minimal $k$ for which we prove \NP-completeness. In the second row choose minimal $n \geq 4$ so that $(n-3)p\geq q$.}\label{t-thetable}
\vspace*{-0.7cm}
\label{tbl:main}
\end{table}
\end{center}

\vspace*{-0.7cm}
\noindent
The proof follows by case analysis as per Table~\ref{tbl:main}, where the corresponding section for each of the subresults is specified. We are able to reduce to the case that $\mathrm{gcd}(p,q)=1$, due to the forthcoming Lemma~\ref{lem:gcd}. We prove \NP-hardness by reduction from graph $3$-colouring and several satisfiability variants. These latter are known to be \NP-hard from Schaefer's classification \cite{Sc78}.
Each section begins with a theorem detailing the relevant \NP-completeness.
The case $p=q=0$ is trivial (never use more than one colour) and is therefore omitted. 
Our hardness proofs involve gadgets that have certain common features, for example, the vertex-variable gadgets are generally star-like. For one case, we have a computer-assisted proof (as we will explain in detail). 

By
Theorem~\ref{thm:hauptsatz} we obtain a complete classification of {\sc $L(p,q)$-Edge-Labelling}.

\begin{corollary}
For all $p,q\geq 0$ except $p=q=0$, {\sc $L(p,q)$-Edge-Labelling} is \NP-complete.
\end{corollary}

\noindent
Note that {\sc $L(p,q)$-Edge-Labelling} is equivalent to {\sc $L(p,q)$-Labelling} for line graphs (the line graph of a graph $G$ has vertex set $E(G)$ and two vertices $e$ and~$f$ in it are adjacent if and only if $e$ and $f$ are adjacent edges in $G$). 
Hence, we obtain another dichotomy for {\sc $L(p,q)$-Labelling} under input restrictions, besides the ones for trees~\cite{FGK08}
and if $p>q$, (planar) bipartite graphs~\cite{JKM09}.

\begin{corollary}\label{c-line}
For all $p,q\geq0$ except $p=q=0$, {\sc $L(p,q)$-Labelling} is \NP-complete for the class of line graphs.
\end{corollary}

{\color{black}
\paragraph{Related work}
An extended abstract of this paper, omitting numerous proofs, appeared at The 16th International Conference and Workshops on Algorithms and Computation (WALCOM) 2022 \cite{WALCOM}.
}

\section{Preliminaries}\label{s-pre}

We use the terms colouring and labelling interchangeably. A special role will be played by 
the \emph{extended $n$-star} (especially for $n=4$). This is a graph built from an $n$-star $K_{1,n}$ by subdividing each edge (so it becomes a path of length $2$). 
Instead of referring to the problem as {\sc $L(p,q)$-Labelling} (or {\sc $L(h,k)$-Labelling}) 
we will use {\sc $L(a,b)$-Labelling} to free these other letters for alternative uses.

The following lemma is folklore and applies equally to the vertex- or edge-labelling problem. Note that $\gcd(0,b)=b$.
\begin{lemma}
Let $\gcd(a,b)=d>1$. Then the identity is a polynomial time reduction from 
{\sc $L(\nicefrac{a}{d},\nicefrac{b}{d})$-(Edge)-$k$-Labelling} to {\sc $L(a,b)$-(Edge)-$kd$-Labelling}.
\label{lem:gcd}
\end{lemma}
%
This result and the known \NP-completeness of  {\sc Edge-$3$-Colouring}~\cite{Ho81} imply:

\begin{corollary}
For all $a>0$, {\sc $L(a,0)$-Edge-$3a$-Labelling} is \NP-complete.
\end{corollary}

\section{Case \boldmath{$a=0$ and $b>0$}}\label{sec:0-1}	

By Lemma~\ref{lem:gcd} we only have to consider $a=0$ and $b=1$.

	\begin{theorem}
		The problem {\sc $L(0,1)$-Edge-$3$-Labelling} is \NP-complete.
		\label{thm:0-1}
	\end{theorem}
	Let us use colours $\{0,1,2\}$. Our \NP-hardness proof involves a reduction from 3-COL but we retain the nomenclature of variable gadget and clause gadget (instead of vertex gadget and edge gadget) in deference to the majority of our other sections. Our variable gadget consists of a triangle attached on one of its vertices to a leaf vertex of a star. Our clause gadget consists of a bull, each of whose pendant edges (vertices of degree $1$) has an additional pendant edge added (that is, they are subdivided). This is equivalent to a triangle with a path of length $2$ added to each of two of the three vertices. We draw our variable gadget in Figure~\ref{fig:0-1-variable} and our clause gadget in Figure~\ref{fig:0-1-clause}.
	\begin{lemma}
	In any valid $L(0,1)$-edge-$3$-labelling of the variable gadget, each of the pendant edges must be coloured the same.
	\label{lem:0-1-variable}
	\end{lemma}
	\begin{proof}
	Each of the edges in the triangle must be coloured distinctly as there is a path of length two from each to any other (by this we mean with a single edge in between, though they are also adjacent). Suppose the triangle edge that has two nodes of degree $2$ in the variable gadget is coloured $i$. It is this colour that must be used for all of the pendant edges. The remaining edge may be coloured by anything from $\{0,1,2\}\setminus \{i\}$. However, we will always choose the option $i-1 \bmod 3$. 
	\end{proof}
		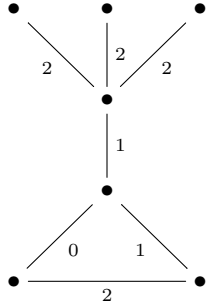
\begin{figure}
\vspace*{-0.3cm}
		\begin{center}
\[
\xymatrix{
\bullet & \bullet &  \bullet \\
&  \bullet  \ar@{-}[u]_{2}  \ar@{-}[ur]_{2}  \ar@{-}[ul]^{2} & &  \\
&  \bullet  \ar@{-}[u]_{1} & &  \\
\bullet \ar@{-}[ur]_{0} \ar@{-}[rr]_{2} & & \bullet \ar@{-}[ul]^{1} \\
}
\]		
		\end{center}
		\caption{The variable gadget for Theorem~\ref{thm:0-1}.}
		\label{fig:0-1-variable}
\vspace*{-0.3cm}
		\end{figure}

	\begin{figure}
		\begin{center}
$
\xymatrix{
& \bullet  \ar@{-}[dd]^{2}  \ar@{-}[ld]^{1} & \bullet  \ar@{-}[l]_{1} &  \bullet  \ar@{-}[l]_{0}  \\
\bullet & & \\
&\bullet  \ar@{-}[lu]_{0} & \bullet  \ar@{-}[l]_{2} &  \bullet  \ar@{-}[l]_{1}  \\
}
$
\hspace{1.5cm}
$
\xymatrix{
& \bullet  \ar@{-}[dd]^{2}  \ar@{-}[ld]^{1} & \bullet  \ar@{-}[l]_{1} &  \bullet  \ar@{-}[l]_{0}  &  \bullet  \ar@{--}[l]_{2} \\
\bullet & & & \\
&\bullet  \ar@{-}[lu]_{0} & \bullet  \ar@{-}[l]_{2} &  \bullet  \ar@{-}[l]_{1}  &  \bullet  \ar@{--}[l]_{0} \\
}
$
		\end{center}
		\caption{The clause gadget for Theorem~\ref{thm:0-1} (left) drawn also together with its interface with a variable gadget (right). The dashed line is an inner edge of the variable gadget.}
		\label{fig:0-1-clause}
		\end{figure}
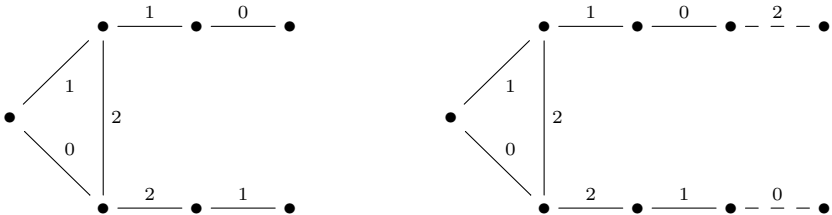
\begin{lemma}
	In any valid $L(0,1)$-edge-$3$-labelling of the clause gadget, the two pendant edges must be coloured distinctly.
	\label{lem:0-1-clause}
	\end{lemma}
	\begin{proof}
	Each of the edges in the triangle must be coloured distinctly as there is a path of length two from each to any other. Suppose the triangle edge that has two nodes of degree $3$ in the clause gadget is coloured (w.l.o.g.) $2$. The remaining edges in the triangle must be given $0$ and $1$, in some order. This then determines the colours of the remaining edges and enforces that the two pendant edges must be coloured distinctly. However, suppose we had started first by colouring distinctly the pendant edges. We could then choose a colouring of the remaining edges of the clause gadget so as to enforce the property that, if a pendant edge is coloured $i$, then its neighbour (in the clause gadget) is coloured $i+1 \bmod 3$. This is the colouring we will always choose. 
	\end{proof}
We are now ready to prove Theorem~\ref{thm:0-1}.
\begin{proof}[Proof of Theorem~\ref{thm:0-1}.]	
We reduce from 3-COL. Let $G$ be an instance of 3-COL involving $n$ vertices and $m$ edges. Let us explain how to build an instance $G'$ for {\sc $L(0,1)$-Edge-$3$-Labelling}. Each particular vertex may only appear in at most $m$ edges (its degree), so for each vertex we take a copy of the variable gadget which has $m$ pendant edges. For each edge of $G$ we use a clause gadget to unite an instance of these pendant edges from the corresponding two variable gadgets. We use each pendant edge from a variable gadget in at most one clause gadget. We identify the pendant edge of a variable gadget with a pendant edge from a clause gadget so as to form a path from one to the other. We claim that $G$ is a yes-instance of 3-COL iff $G'$ is a yes-instance of {\sc $L(0,1)$-Edge-$3$-Labelling}.

(Forwards.) Take a proper $3$-colouring of $G$ and induce these colours on the pendant edges of the corresponding variable gadgets. Distinct colours on pendant edges can be consistently united in a clause gadget since we choose, for a pendant edge coloured $i$: $i-1 \bmod 3$ for its neighbour in the variable gadget, and $i+1 \bmod 3$ for its neighbour in the clause gadget.

(Backwards.) From a valid $L(0,1)$-edge-$3$-labelling of $G'$, we infer a $3$-colouring of $G$ by reading the pendant edge labels from the variable gadget of the corresponding vertex. The consistent labelling of each vertex follows from Lemma~\ref{lem:0-1-variable} and the fact that it is proper follows from Lemma~\ref{lem:0-1-clause}. 
\end{proof}

\section{Case \boldmath{$2 \leq \frac{b}{a}$}}\label{a-2}

In the case $2 \leq \frac{b}{a}$, we can no longer get away with just an extended $4$-star on which to base our variable gadget (as we did in Section~\ref{sec:1-to-2}). We need to move to higher degree. On the other hand, we will be able to dispense with the pendant $5$-stars.
	\begin{theorem}
	If $2\leq \frac ba$, let $n\geq 4$ be such that $(n-3)a\geq b$  then problem 	
	{\sc $L(a,b)$-Edge-$((n-1)a+b+1)$-Labelling} is \NP-complete.
	\label{thm:beyond-2}
	\end{theorem}
We will need the following lemma.
	\begin{lemma}
	Let $2\leq \frac ba$ and let $n\geq 4$ be such that \textcolor{black}{$(n-3)a\geq b$}. In any valid $L(a,b)$-edge-$((n-1)a+b+1)$-labelling of the extended $n$-star, either all pendant edges are coloured in the interval 	$\{(n-2)a+b,\ldots,(n-1)a+b\}$ or all pendant edges are coloured in the interval $\{0,\ldots,a\}$.
	\label{lem:beyond-2}
	\end{lemma}
	\begin{proof}
	Suppose some pendant edge is coloured by $l'$ in $\{a+1,\ldots,(n-2)a+b-1\}$. Consider the $n-1$ inner edges at distance $2$ from it. Reading their labels in ascending order there must be a jump of at least \textcolor{black}{$2b\geq a+b+1$ at some point unless the lowest label is itself $a+b+1$}. But now we have run out of labels, because $(n-2)a+(a+b+1)>(n-1)a+b$ which is the last label.
	
	Suppose now that some pendant edge is coloured by $l'_1$ in $\{0,\ldots,a\}$ and another pendant edge is coloured by $l'_2$ in $\{(n-2)a+b,\ldots,(n-1)a+b\}$. It is now not possible to choose $n-2$ labels to complete the opposing inner edges, because $l_1$ and $l_2$ \textcolor{black}{(inner edges adjacent to outer edges with labels $l'_1$ and $l'_2$, respectively)} together must remove more than $b\geq 2a$  possibilities for labels at both the top and the bottom of the order. Using $2b> b+2a$, this leaves no more than $(n-3)a$ which is not enough space for $n-2$ labels spaced by $a$ in the $n-2$ inner edges.  
	
		Finally, we note a valid colouring of the form $0,\ldots,(n-1)a$ for the inner edges of the extended $n$-star, with 	$\{(n-2)a+b,\ldots,(n-1)a+b\}$ enforced on the pendant edges (and the whole range from $\{(n-2)a+b,\ldots,(n-1)a+b\}$ \textcolor{black}{is} possible adjacent to the label $(n-1)a$). The other regime comes from order-inverting the colours.  
	\end{proof}
The stipulation $(n-3)a\geq b$ plays no role in the previous lemma. It is needed in order to chain together extended $n$-stars to form the \emph{variable gadget} whose construction we now explain. The variable gadget is made from a series of extended $n$-stars joined in a chain. They can join to one another in a path running from one's inner star edge labelled $0$ to another's inner star edge labelled $(n-2)a$. In this fashion, the inner star edge labelled $(n-1)a$ is free for the (top) pendant edge that acts as the point of contact for clauses. This inner star edge may sometimes need to be labelled $(n-2)a$ (\mbox{cf.} Figure~\ref{fig:beyond-2-clause}) in which case the other inner star edge labelled $(n-3)a$ will be needed to perform the chaining.
	\begin{figure}
		\begin{center}
			\includegraphics[scale=.8]{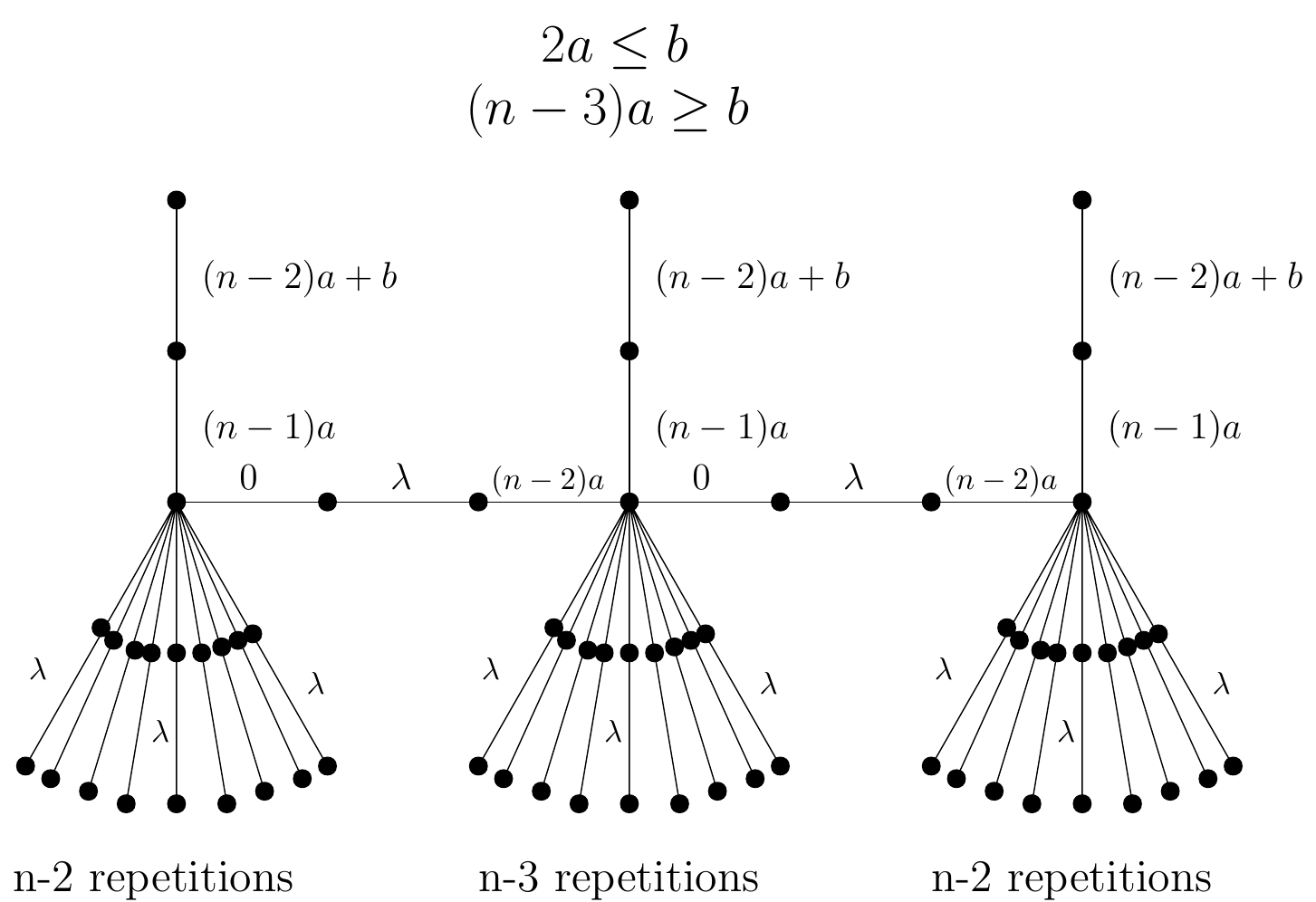}
		\end{center}
		\caption{The variable gadget for Theorem~\ref{thm:beyond-2}. The pendant edges drawn on the top will be involved in clauses gadget and each of these three edges can be coloured with anything from $\{(n-2)a+b,\ldots,(n-1)a+b\}$.}
		\label{fig:beyond-2}
		\end{figure}
In the following lemma, the designation \emph{top} is with reference to the drawing in Figure~\ref{fig:beyond-2}.
\begin{lemma}
\textcolor{black}{Let $n\geq 4$ be such that $(n-3)a\geq b$.} Any valid $L(a,b)$-edge-$(n-1)a+b+1$-labelling of a variable gadget is such that the top pendant edges are all coloured from precisely one of the sets $\{0,\ldots,a\}$ and $\{(n-2)a+b,\ldots,(n-1)a+b\}$. Moreover, any colouring of the top pendant edges from one of these sets is valid.
\label{lem:beyond-2-variable}
\end{lemma}
The clause gadget will be nothing more than a $3$-star (a claw) which is formed from a new vertex uniting three (top) pendant edges from their respective variable gadgets. The following is clear.
\begin{lemma}
\textcolor{black}{Let $n\geq 4$ be such that $(n-3)a\geq b$.} A clause gadget is in a valid $L(a,b)$-edge-$(n-1)a+b+1$-labelling in the case where two of its edges are coloured $0,a$ and the third $(n-1)a+b$; or two of its edges are coloured $(n-2)a+b,(n-1)a+b$ and the third $0$. If all three edges come from only one of the regimes $\{0,\ldots,a\}$ and $\{(n-2)a+b,\ldots,(n-1)a+b\}$, it cannot be in a valid $L(a,b)$-edge-$(n-1)a+b+1$-labelling.
\label{lem:beyond-2-clause}
\end{lemma}
We are now ready to prove Theorem~\ref{thm:beyond-2}.
\begin{proof}[Proof of Theorem~\ref{thm:beyond-2}.]	
We reduce from (monotone) NAE-3-SAT. Choose $n$ such that \textcolor{black}{$(n-3)a\geq b$}. Let $\Phi$ be an instance of NAE-3-SAT involving $N$ occurrences of (not necessarily distinct) variables and $m$ clauses. Let us explain how to build an instance $G$ for {\sc $L(a,b)$-Edge-$(n-1)a+b+1$-Labelling}. Each particular variable may only appear at most $N$ times, so for each variable we take a copy of the variable gadget which is $N$ extended $n$-stars chained together. Each particular instance of the variable belongs to one of the free (top) pendant edges of the variable gadget. For each clause of $\Phi$ we use a $3$-star to unite an instance of these free (top) pendant edges from the corresponding variable gadgets. Thus, we add a single vertex for each clause, but no new edges (they already existed in the variable gadgets). We claim that $\Phi$ is a yes-instance of NAE-3-SAT if and only if $G$ is a yes-instance of {\sc $L(a,b)$-Edge-$(n-1)a+b+1$-Labelling}.

(Forwards.) Take a satisfying assignment for $\Phi$. Let the range $\{0,\ldots,a\}$ represent true and the range $\{(n-2)a+b,\ldots,(n-1)a+b\}$ represent false. This gives a valid labelling of the inner vertices in the extended $n$-stars, as exemplified in Figure~\ref{fig:beyond-2}. In each clause, either there are two instances of true and one of false; or the converse. Let us explain the case where the first two variable instances are true and the third is false (the general case can easily be garnered from this). Colour the (top) pendant edge associated with the first variable as $0$, the second variable $a$ and the third variable $(n-1)a+b$. Plainly these can be consistently united in a claw by the new vertex that appeared in the clause gadget. We draw the situation in Figure~\ref{fig:beyond-2-clause} to demonstrate that this will not introduce problems at distance $2$. Thus, we can see this is a valid $L(a,b)$-edge-$(n-1)a+b+1$-labelling of $G$.

(Backwards.) From a valid $L(a,b)$-edge-$(n-1)a+b+1$-labelling of $G$, we infer an assignment $\Phi$ by reading, in the variable gadget, the  range $\{0,\ldots,a\}$ as true and the range $\{(n-2)a+b,\ldots,(n-1)a+b\}$ as false. The consistent valuation of each variable follows from Lemma~\ref{lem:beyond-2-variable} and the fact that it is in fact not-all-equal follows from Lemma~\ref{lem:beyond-2-clause}. 
\end{proof}

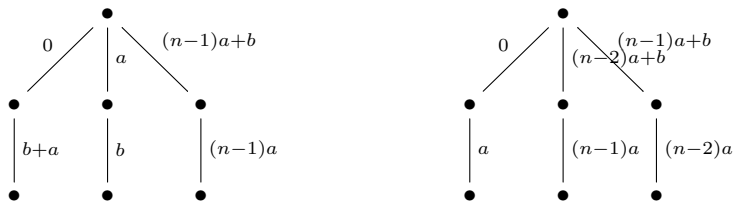
\begin{figure}
$
\xymatrix{
& \bullet  \ar@{-}[dl]_{0} \ar@{-}[d]^{a}  \ar@{-}[dr]^{(n-1)a+b}& \\
\bullet \ar@{-}[d]^{b+a} & \bullet \ar@{-}[d]^{b} & \bullet \ar@{-}[d]^{(n-1)a}\\
\bullet & \bullet & \bullet \\
}
$
\hspace{2cm}
$
\xymatrix{
& \bullet  \ar@{-}[dl]_{0} \ar@{-}[d]^{(n-2)a+b}  \ar@{-}[dr]^{(n-1)a+b}& \\
\bullet \ar@{-}[d]^{a} & \bullet \ar@{-}[d]^{(n-1)a} & \bullet \ar@{-}[d]^{(n-2)a}\\
\bullet & \bullet & \bullet \\
}
$
\caption{The clause gadget and its interface with the variable gadgets (where we must consider distance $2$ constraints). Both possible evaluations for not-all-equal are depicted. Note the difference $(n-2)a+b - (n-1)a = b-a >a$.}
\label{fig:beyond-2-clause}
\end{figure}

\section{Case \boldmath{$1< \frac{b}{a} \leq 2$}}
\label{sec:1-to-2}	

In this section we prove the following result.

	\begin{theorem}
		If $1<\frac ba \leq 2$, the problem {\sc $L(a,b)$-Edge-$(5a+1)$-Labelling} is \NP-complete.
		\label{thm:1-to-2}
	\end{theorem}

\noindent
We proceed by a reduction from (monotone) NAE-3-SAT.
This case is relatively simple as the variable gadget is built from a series of extended $4$-stars chained together, where each has a pendant $5$-star to enforce some benign property. We will use colours from the set $\{0,\ldots,5a\}$. 
	
	\begin{lemma}
	Let $1<\frac ba \leq 2$. In any valid $L(a,b)$-edge-$(5a+1)$-labelling of the extended $4$-star, if one pendant edge is coloured $0$ then all pendant edges are coloured in the interval $\{0,\ldots,a\}$; and if one pendant edge is coloured $5a$ then all pendant edge are coloured in the interval $\{4a,\ldots,5a\}$.
	\label{lem:1-to-2}
	\end{lemma}
	\begin{proof}
	Suppose some pendant edge is coloured by $0$ and another pendant is coloured by $l' \notin \{0,\ldots,a\}$. There are four inner edges of the star that are at distance $1$ or $2$ from these, and one another \textcolor{black}{(indeed, they are at distance $1$ from one another)}. If $l'<2a$, then at least $2a$ labels are ruled out, which does not leave enough possibilities for the inner edges to be labelled in (at best) $\{2a+1,\ldots,5a\}$. If $l'\geq 2a$, then it is not possible to use labels for the inner edges that are all strictly above $l'$. It is also not possible to use labels for the inner edges that are all strictly below \textcolor{black}{$l'$}. In both cases, at least $2a$ labels are ruled out. Thus the labels, read in ascending order, must start no lower than $a$ and have a jump of $2a$ at some point. It follows they are one of: $a,3a,4a,5a$; or $a,2a,4a,5a$; or $a,2a,3a,5a$. This implies that $l'$ is itself a multiple of $a$ (whichever one was omitted in the given sequence). But now, since $b>a$, there must be a violation of a distance $2$ constraint from $l'$. 
	\end{proof}
\noindent \textcolor{black}{Let us remark that the colourings as restricted in Lemma~\ref{lem:1-to-2} are achievable, and we will use them in the sequel.}

We would like to chain extended $4$-stars together to build our variable gadgets, where the pendant edges represent variables (and enter into clause gadgets) and we interpret one of the regimes $\{0,\ldots,a\}$ and $\{4a,\ldots,5a\}$ as true, and the other as false. However, the extended $4$-star can be validly $L(a,b)$-edge-($5a+1$)-labelled in other ways that we did not yet consider. We can only use Lemma~\ref{lem:1-to-2} if we can force one pendant edge in each extended $4$-star to be either $0$ or $5a$. Fortunately, this is straightforward: take a $5$-star and add a new edge to one of the edges of the $5$-star creating a path of length $2$ from the centre of the star to the furthest leaf. This new edge can only be coloured $0$ or $5a$. In Figure~\ref{fig:1-to-2} we show how to chain together copies of the extended $4$-star, together with pendant $5$-star gadgets at the bottom, to produce many copies of exactly one of the regimes $\{0,\ldots,a\}$ and $\{4a,\ldots,5a\}$. Note that the manner in which we attach the pendant $5$-star only produces a valid $L(a,b)$-edge-($5a+1$)-labelling because $2a \geq b$ (otherwise some distance $2$ constraints would fail). So long as precisely one pendant edge per extended $4$-star is used to encode a variable, then each encoding can realise all labels within each of these regimes, and again this can be seen by considering the pendant edges drawn top-most in Figure~\ref{fig:1-to-2}, which can all be coloured anywhere in $\{4a,\ldots,5a\}$. Let us recap, a \emph{variable gadget} (to be used for a variable that appears in an instance of NAE-3-SAT $m$ times) is built from chaining together $m$ extended $4$-stars, each with a pendant $5$-star, exactly as is depicted in Figure~\ref{fig:1-to-2} for $m=3$.
		\begin{figure}
\vspace*{-0.58cm}
		\begin{center}
\[
\xymatrix{
& &  \bullet  \ar@{-}[d]^{\{4a,\ldots,5a\}} & & & \bullet  \ar@{-}[d]^{\{4a,\ldots,5a\}} & & & \bullet  \ar@{-}[d]^{\{4a,\ldots,5a\}} & & & \\
& &  \bullet  \ar@{-}[d]^{3a} & & & \bullet  \ar@{-}[d]^{3a} & & & \bullet  \ar@{-}[d]^{3a} & & & \\
\bullet  \ar@{-}[r]^{5a} & \bullet  \ar@{-}[r]^{2a} & \bullet  \ar@{-}[r]^{0} 
 & \bullet  \ar@{-}[r]^{5a} & \bullet  \ar@{-}[r]^{2a} & \bullet  \ar@{-}[r]^{0} & \bullet  \ar@{-}[r]^{5a} & \bullet  \ar@{-}[r]^{2a} & \bullet  \ar@{-}[r]^{0} & \bullet  \ar@{-}[r]^{5a} & \bullet \\
& &  \bullet  \ar@{-}[u]_{a} & & & \bullet  \ar@{-}[u]_{a} & & & \bullet  \ar@{-}[u]_{a} & & & \\
& &  \bullet  \ar@{-}[u]_{5a} & & & \bullet  \ar@{-}[u]_{5a} & & & \bullet  \ar@{-}[u]_{5a} & & & \\
& \bullet \ar@{-}[r]_{0} &  \bullet  \ar@{-}[u]_{4a} & \bullet \ar@{-}[l]^{a} &  \bullet \ar@{-}[r]_{0} &  \bullet  \ar@{-}[u]_{4a} & \bullet \ar@{-}[l]^{a} &  \bullet \ar@{-}[r]_{0} &  \bullet  \ar@{-}[u]_{4a} & \bullet \ar@{-}[l]^{a} & \\
& \bullet \ar@{-}[ur]_{2a} &  & \bullet \ar@{-}[ul]^{3a} & \bullet \ar@{-}[ur]_{2a} &  & \bullet \ar@{-}[ul]^{3a} & \bullet \ar@{-}[ur]_{2a} &  & \bullet \ar@{-}[ul]^{3a} & \\
}
\]		

		\end{center}
\vspace*{-0.3cm}
		\caption{Three extended $4$-stars chained together, each with a pendant $5$-star below, to form a variable gadget for Theorem~\ref{thm:1-to-2}. The pendant edges drawn on the top will be involved in clauses gadget and each of these three edges can be coloured with anything from $\{4a,\ldots,5a\}$. If the top pendant edge is coloured $5a$ it may be necessary that the inner star edge below is coloured not $3a$ but $2a$ (\mbox{cf.} Figure~\ref{fig:1-to-2-clause}). This is fine, the chaining construction works when swapping $2a$ and $3a$.}
		\label{fig:1-to-2}
\vspace*{-0.4cm}
		\end{figure}
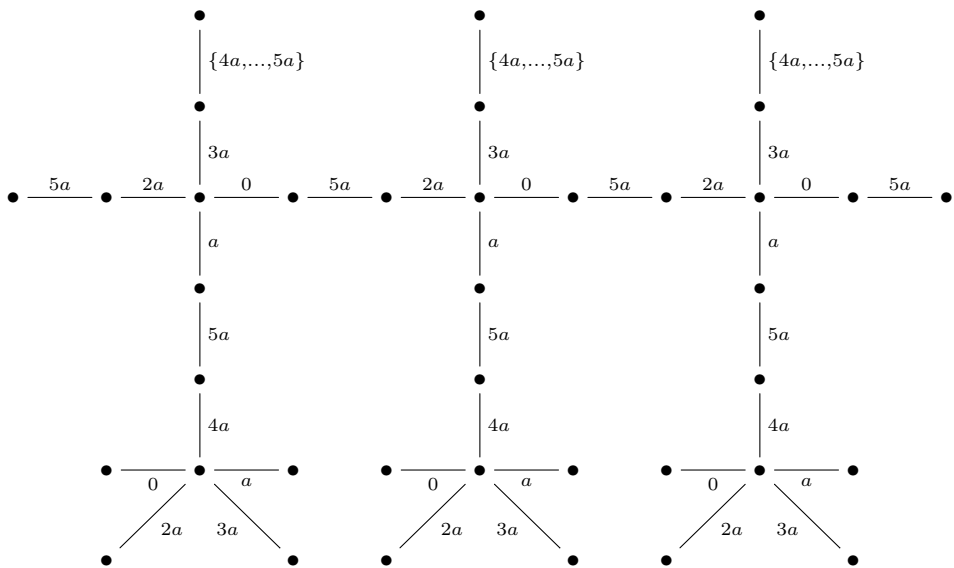
		The following is clear from our construction. The designation \emph{top} is with reference to the drawing in Figure~\ref{fig:1-to-2}. 
In Figure~\ref{fig:1-to-2}, the case drawn corresponds to $\{4a,\ldots,5a\}$, where the case $\{0,\ldots,a\}$ is symmetric.	
\begin{lemma}
Any valid $L(a,b)$-edge-$(5a+1)$-labelling of a variable gadget is such that the top pendant edges are all coloured from precisely one of the sets $\{0,\ldots,a\}$ and $\{4a,\ldots,5a\}$. Moreover, any colouring of the top pendant edges from one of these sets is valid.
\label{lem:1-to-2-variable}
\end{lemma}
The clause gadget will be nothing more than a $3$-star (a claw) which is formed from a new vertex uniting three (top) pendant edges from their respective variable gadgets. The following is clear.
\begin{lemma}
A clause gadget is in a 
valid 
$L(a,b)$-edge-$(5a+1)$-labelling in the case where two of its edges are coloured $0,a$ and the third $5a$; or two of its edges are coloured $4a,5a$ and the third $0$. If all three edges come from only one of the regimes $\{0,\ldots,a\}$ and $\{4a,\ldots,5a\}$, it can not be in a valid $L(a,b)$-edge-$(5a+1)$-labelling.
\label{lem:1-to-2-clause}
\end{lemma}
We are now ready to prove Theorem~\ref{thm:1-to-2}.
\begin{proof}[Proof of Theorem~\ref{thm:1-to-2}.]	
We reduce from (monotone) NAE-3-SAT. Let $\Phi$ be an instance of NAE-3-SAT involving $n$ occurrences of (not necessarily distinct) variables and $m$ clauses. Let us explain how to build an instance $G$ for {\sc $L(a,b)$-Edge-$(5a+1)$-Labelling}. Each particular variable may only appear at most $n$ times, so for each variable we take a copy of the variable gadget which is $n$ extended $4$-stars, each with a pendant $5$-star, chained together. Each particular instance of the variable belongs to one of the free (top) pendant edges of the variable gadget. For each clause of $\Phi$ we use a $3$-star to unite an instance of these free (top) pendant edges from the corresponding variable gadgets. Thus, we add a single vertex for each clause, but no new edges (they already existed in the variable gadgets). We claim that $\Phi$ is a yes-instance of NAE-3-SAT if and only if $G$ is a yes-instance of {\sc $L(a,b)$-Edge-$(5a+1)$-Labelling}.

(Forwards.) Take a satisfying assignment for $\Phi$. Let the range $\{0,\ldots,a\}$ represent true and the range $\{4a,\ldots,5a\}$ represent false. This gives a valid labelling of the inner edges in the extended $4$-stars, as exemplified in Figure~\ref{fig:1-to-2}. In each clause, either there are two instances of true and one of false; or the converse. Let us explain the case where the first two variable instances are true and the third is false (the general case can easily be garnered from this). Colour the (top) pendant edge associated with the first variable as $0$, the second variable $a$ and the third variable $5a$. Plainly these can be consistently united in a claw by the new vertex that appeared in the clause gadget. We draw the situation in Figure~\ref{fig:1-to-2-clause} to demonstrate that this will not introduce problems at distance $2$. Thus, we can see this is a valid $L(a,b)$-edge-$(5a+1)$-labelling of $G$.

(Backwards.) From a valid $L(a,b)$-edge-$(5a+1)$-labelling of $G$, we infer an assignment $\Phi$ by reading, in the variable gadget, the  range $\{0,\ldots,a\}$ as true and the range $\{4a,\ldots,5a\}$ as false. The consistent valuation of each variable follows from Lemma~\ref{lem:1-to-2-variable} and the fact that it is in fact not-all-equal follows from Lemma~\ref{lem:1-to-2-clause}. \qed
\end{proof}

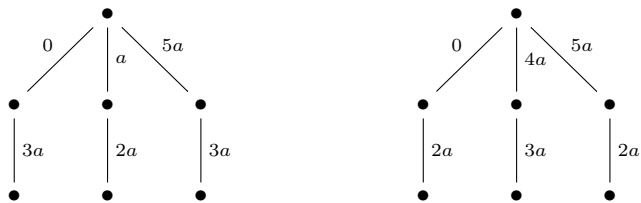
\begin{figure}
$
\xymatrix{
& \bullet  \ar@{-}[dl]_{0} \ar@{-}[d]^{a}  \ar@{-}[dr]^{5a}& \\
\bullet \ar@{-}[d]^{3a} & \bullet \ar@{-}[d]^{2a} & \bullet \ar@{-}[d]^{3a}\\
\bullet & \bullet & \bullet \\
}
$
\hspace{2cm}
$
\xymatrix{
& \bullet  \ar@{-}[dl]_{0} \ar@{-}[d]^{4a}  \ar@{-}[dr]^{5a}& \\
\bullet \ar@{-}[d]^{2a} & \bullet \ar@{-}[d]^{3a} & \bullet \ar@{-}[d]^{2a}\\
\bullet & \bullet & \bullet \\
}
$
\caption{The clause gadget and its interface with the variable gadgets (where we must consider distance $2$ constraints). Both possible evaluations for not-all-equal are depicted.}
\label{fig:1-to-2-clause}
\end{figure}


		\section{Case \boldmath{$\frac 23< \frac{b}{a} < 1$}}
		\label{sec:3-col}

In this section we prove the following result.

		\begin{theorem}
		If $\frac 23<\frac ba < 1$, then the problem {\sc $L(a,b)$-Edge-$(3a+b+1)$-Labelling} is \NP-complete.
		\label{thm:3-col}
	\end{theorem}
	
\begin{figure}
\begin{center}
	\includegraphics[width=1.03\linewidth]{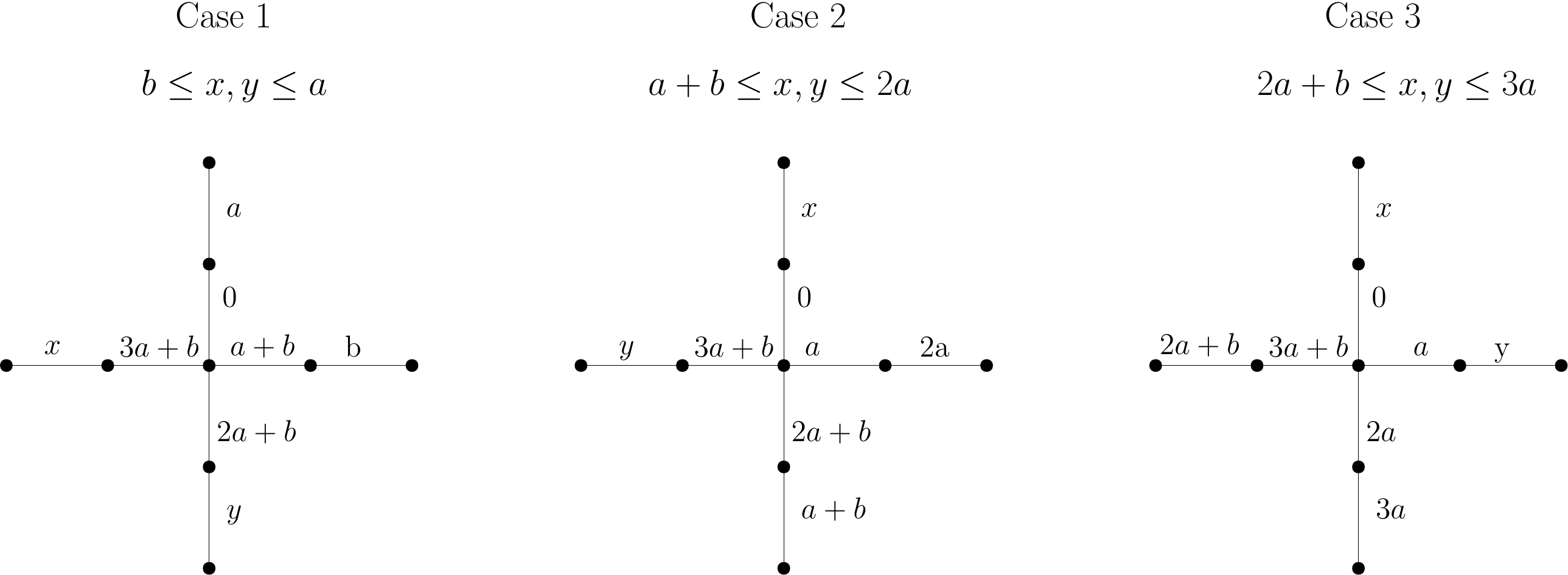}
\end{center}
\caption{The regimes of Theorem~\ref{thm:3-col}.}
\label{fig:3-col-reg}
\end{figure}
	
	The regimes of the following lemma are drawn in Figure~\ref{fig:3-col-reg}. 
	\begin{lemma}
	Let $1<\frac{a}{b}< \frac{3}2$. In an $L(a,b)$-edge-$(3a+b+1)$-labelling $c$ of the extended 4-star, there are three regimes for the pendant edges. The first is $\{b,\ldots,a\}$, the second is $\{2a+b,\ldots,3a\}$, and the third is $\{a+b,\ldots,2a\}$.
		\label{lem:3-col}
	\end{lemma}

	\begin{proof}
		In a valid $L(a,b)$-edge-$(3a+b+1)$-labelling,
		we note $c_1<c_2<c_3<c_4$ the colours of the 4 edges in the middle of the extended $4$-star, and $l_1,l_2,l_3,l_4$ the colours of the pendant edges such that $l_i$ is the colour of the pendant edge connected to the edge of colour $c_i$.

\medskip
\noindent \textbf{Claim 1}. For all $i$, $c_1 < l_i < c_4.$
\medskip

	We only have to prove one inequality, as the other one is obtained by symmetry.
	If $l_i \leq c_1$ (bearing in mind also $b<a$), we have: $$3a+b\geq c_4-l_i=(c_1-l_i)+(c_2-c_1)+(c_3-c_2)+(c_4-c_3)\geq 3a+b.$$
	So $(c_1,c_2,c_3,c_4)=(b,a+b,2a+b,3a+b)$, but $a>b$ so there is no possible value for $l_1$, which is not possible.
	So $c_1<l_i$, and by symmetry $l_i<c_4$.

\medskip
\noindent \textbf{Claim 2}. There exists $i \in \{1,2,3\}$ such that $c_{i+1}-c_{i} \geq a+b$.
\medskip

	We suppose the contrary. We have proved $c_1<l_2,l_3<c_4$.
	If $l_2<c_2$, then $c_2-c_1=c_2-l_2+l_2-c_1\geq a+b$, 
	impossible.
	If $c_2<l_2<c_3$, then $c_3-c_2=c_3-l_2+l_2-c_2\geq a+b$, impossible.
	So $c_3<l_2<c_4$.
	 Symmetrically, we obtain $c_1<l_3<c_2$.
	So $c_1 < l_3 < c_2 < c_3< l_2 < c_4$, and we get:
	%
	$c_4-c_1 
	\geq (l_3-c_1)+(c_2-l_3)+(c_3-c_2)+(l_2-c_3)+(c_4-l_2) 
	\geq 4b+a  
	>3a+b$, {which is not possible}.

Now we are in a position to derive the lemma, with the three regimes coming from the three possibilities of Claim 2. If $i=1$, then the inner edges of the star are $0,a+b,2a+b,3a+b$ and the pendant edges come from  $\{b,\ldots,a\}$. If $i=2$, then the  inner edges of the star are $0,a,2a+b,3a+b$ and the pendant edges come from $\{a+b,\ldots,2a\}$. If $i=3$, then the  inner edges of the star are $0,a,a+b,3a+b$ and the pendant edges come from $\{2a+b,\ldots,3a\}$. 
	\end{proof}
The \emph{variable gadget}  may be taken as a series of extended $4$-stars chained together. In the following, the ``top'' pendant edges refer to one of the two free pendant edges in each extended $4$-star (not involved in the chaining together). The following is a simple consequence of Lemma~\ref{lem:3-col} and is depicted in Figure~\ref{fig:3-col}.
\begin{lemma}
Any valid $L(a,b)$-edge-$(3a+b+1)$-labelling of a variable gadget is such that the top pendant edges are all coloured from precisely one of the sets $\{b,\ldots,a\}$, $\{a+b,\ldots,2a\}$ or $\{2a+b,\ldots,3a\}$. Moreover, any colouring of the top pendant edges from one of these sets is valid.
\label{lem:3-col-variable}
\end{lemma}
	\begin{figure}
		\begin{center}
\[
\xymatrix{
& &  \bullet  \ar@{-}[d]^{b} & & & \bullet  \ar@{-}[d]^{b} & & & \bullet  \ar@{-}[d]^{b} & & & \\
& &  \bullet  \ar@{-}[d]^{a+b} & & & \bullet  \ar@{-}[d]^{a+b} & & & \bullet  \ar@{-}[d]^{a+b} & & & \\
\bullet  \ar@{-}[r]^{\textcolor{black}{a}} & \bullet  \ar@{-}[r]^{0} & \bullet  \ar@{-}[r]^{2a+b} 
 & \bullet  \ar@{-}[r]^{\textcolor{black}{a}} & \bullet  \ar@{-}[r]^{0} & \bullet  \ar@{-}[r]^{2a+b} & \bullet  \ar@{-}[r]^{\textcolor{black}{a}} & \bullet  \ar@{-}[r]^{0} & \bullet  \ar@{-}[r]^{2a+b} & \bullet  \ar@{-}[r]^{\textcolor{black}{a}} & \bullet \\
& &  \bullet  \ar@{-}[u]_{3a+b} & & & \bullet  \ar@{-}[u]_{3a+b} & & & \bullet  \ar@{-}[u]_{3a+b} & & & \\
& &  \bullet  \ar@{-}[u]_{b} & & & \bullet  \ar@{-}[u]_{b} & & & \bullet  \ar@{-}[u]_{b} & & & \\
}
\]		

		\end{center}
		\caption{Three extended $4$-stars chained together, to form a variable gadget for Theorem~\ref{thm:3-col}. The pendant edges drawn on the top will be involved in clauses gadget. Suppose the top pendant edges are coloured $b$ (as is drawn). In order to fulfill distance $2$ constraints in the clause gadget, we may need the inner star vertices adjacent to them to be coloured not always $a+b$ (for example, if that pendant edge $b$ is adjacent in a clause gadget to another edge coloured $a+b$). This is fine, the chaining construction works when swapping inner edges $a+b$ and $\textcolor{black}{3a+b}$ wherever necessary.}
		\label{fig:3-col}
		\end{figure}
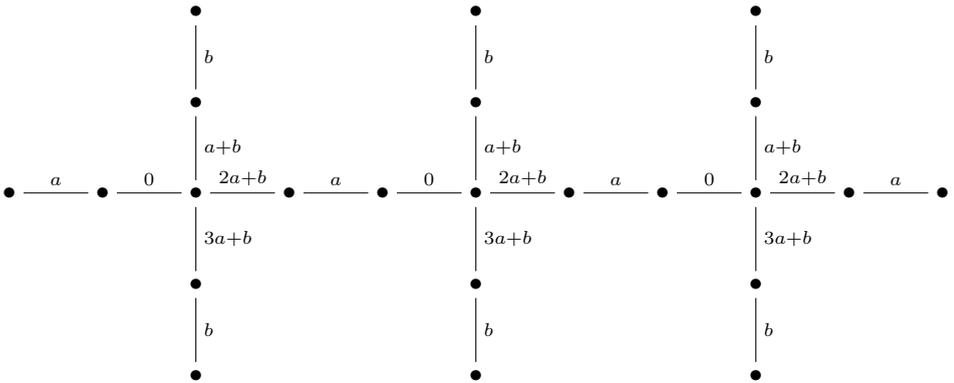
The clause gadget will be nothing more than a $2$-star (a path) which is formed from a new vertex uniting two (top) pendant edges from their respective variable gadgets. The following is clear.
\begin{lemma}
A clause gadget is in a valid $L(a,b)$-edge-$(3a+b+1)$-labelling in the case where its edges are coloured distinctly. If they are coloured the same, then it can not be in a valid $L(a,b)$-edge-$(3a+b+1)$-labelling. 
\label{lem:3-col-clause}
\end{lemma}
We are now ready to prove Theorem~\ref{thm:3-col}.

\begin{proof}[Proof of Theorem~\ref{thm:3-col}.]	
We reduce from 3-COL. Let $G$ be an instance of 3-COL involving $n$ vertices and $m$ edges. Let us explain how to build an instance $G'$ for {\sc $L(a,b)$-Edge-$(3a+b+1)$-Labelling}. Each particular vertex may only appear in at most $m$ edges ($m$ is an upper ground on its degree), so for each vertex we take a copy of the variable gadget which is $m$ extended $4$-stars chained together. Each particular instance of the vertex belongs to one of the free (top) pendant edges of the variable gadget. For each edge of $G$ we use a $2$-star to unite an instance of these free (top) pendant edges from the corresponding two variable gadgets. Thus, we add a single vertex for each edge of $G$, but no new edges in $G'$ (they already existed in the variable gadgets). We claim that $G$ is a yes-instance of 3-COL if and only if $G'$ is a yes-instance of {\sc $L(a,b)$-Edge-$(3a+b+1)$-Labelling}.

(Forwards.) Take a proper $3$-colouring of $G$ and induce these pendant edge labels on the corresponding variable gadgets according to the three regimes of Lemma~\ref{lem:3-col}. For example, map colours $1$, $2$, $3$ to $b, a+b, 2a+b$. Plainly distinct pendant edge labels can be consistently united in a $2$-claw by the new vertex that appeared in the clause gadget. Thus, we can see this is a valid $L(a,b)$-edge-$(3a+b+1)$-labelling of $G'$.

(Backwards.) From a valid $L(a,b)$-edge-$(3a+b+1)$-labelling of $G'$, we infer a $3$-colouring of $G$ by reading the pendant edge labels from the variable gadget of the corresponding vertex and mapping these to their corresponding regime. The consistent valuation of each variable follows from Lemma~\ref{lem:3-col-variable} and the fact that it is proper (not-all-equal) follows from Lemma~\ref{lem:3-col-clause}. 
\end{proof}

\section{Case \boldmath{$\frac{b}{a} = \frac{2}{3}$}}
\label{sec:1-in-3}
 In light of Lemma~\ref{lem:gcd}, it suffices to find $k$ so that {\sc $L(3,2)$-Edge-$k$-Labelling} is \NP-hard.

	\begin{theorem}
		The problem {\sc $L(3,2)$-Edge-$12$-Labelling} problem is \NP-complete.
		\label{thm:1-in-3}
	\end{theorem}
We use the colours $\{0,\ldots,11\}$. The following can be verified by hand; we used a computer.\footnotemark
\footnotetext{The Python program for checking this can be found at \texttt{https://perso.ens-lyon.fr/gaetan.berthe/edge-labelling/labelling-code.py}. Use, e.g.: \texttt{extended\_four\_star=[(0,1),(0,2),(0,3),(0,4),(1,5),(2,6),(3,7),(4,8)]} with \texttt{plotPoss(extended\_four\_star, 3, 2, 12,  \{(1,5):[2]\})} to find $2: \{2,3,9\}$.}
		\begin{lemma}
			In a valid $L(3,2)$-edge-$12$-labelling of the extended 4-star, the possible labels of the three other pendant edges after one label is fixed are given in the following dictionary: 
			\begin{align*}
				2:& \{2, 3, 9\}\\3:& \{2, 3\}\\ 5:& \{5, 6\}\\6:& \{5, 6\}\\ 8:& \{8, 9\}\\ 9:& \{2, 8, 9\}
			\end{align*}
			\label{lem:1-in-3}
		\end{lemma}
	 The \emph{variable gadget} may be taken as a series of extended $4$-stars chained together. In the following, the ``top'' pendant edges refer to one of the two free pendant edges in each extended $4$-star (not involved in the chaining together). 
\begin{lemma}
Any valid $L(3,2)$-edge-$12$-labelling of a variable gadget is such that the top pendant edges are all coloured from precisely one of the sets $\{5,6\}$ or $\{2,3,8,9\}$. Moreover, any colouring of the top pendant edges from one of these sets is valid.
\label{lem:1-in-3-variable}
\end{lemma}
	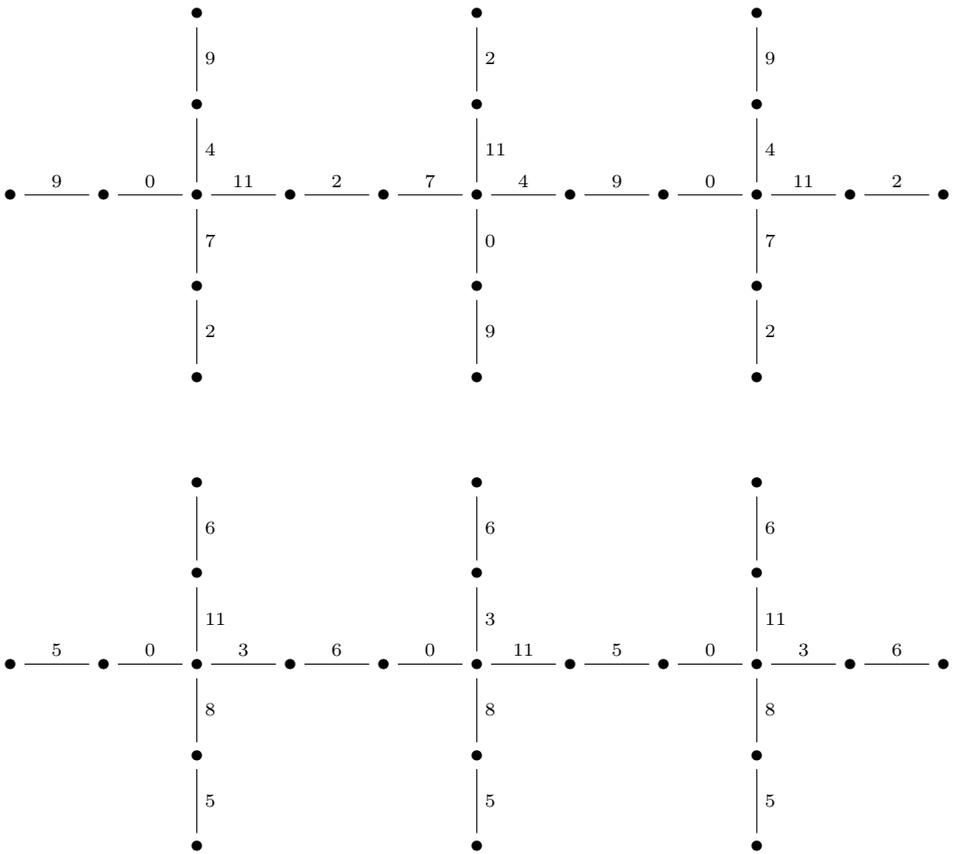
\begin{figure}
		\begin{center}
\[
\xymatrix{
& &  \bullet  \ar@{-}[d]^{9} & & & \bullet  \ar@{-}[d]^{2} & & & \bullet  \ar@{-}[d]^{9} & & & \\
& &  \bullet  \ar@{-}[d]^{4} & & & \bullet  \ar@{-}[d]^{11} & & & \bullet  \ar@{-}[d]^{4} & & & \\
\bullet  \ar@{-}[r]^{9} & \bullet  \ar@{-}[r]^{0} & \bullet  \ar@{-}[r]^{11} 
 & \bullet  \ar@{-}[r]^{2} & \bullet  \ar@{-}[r]^{7} & \bullet  \ar@{-}[r]^{4} & \bullet  \ar@{-}[r]^{9} & \bullet  \ar@{-}[r]^{0} & \bullet  \ar@{-}[r]^{11} & \bullet  \ar@{-}[r]^{2} & \bullet \\
& &  \bullet  \ar@{-}[u]_{7} & & & \bullet  \ar@{-}[u]_{0} & & & \bullet  \ar@{-}[u]_{7} & & & \\
& &  \bullet  \ar@{-}[u]_{2} & & & \bullet  \ar@{-}[u]_{9} & & & \bullet  \ar@{-}[u]_{2} & & & \\
}
\]		

\[
\xymatrix{
& &  \bullet  \ar@{-}[d]^{6} & & & \bullet  \ar@{-}[d]^{6} & & & \bullet  \ar@{-}[d]^{6} & & & \\
& &  \bullet  \ar@{-}[d]^{11} & & & \bullet  \ar@{-}[d]^{3} & & & \bullet  \ar@{-}[d]^{11} & & & \\
\bullet  \ar@{-}[r]^{5} & \bullet  \ar@{-}[r]^{0} & \bullet  \ar@{-}[r]^{3} 
 & \bullet  \ar@{-}[r]^{6} & \bullet  \ar@{-}[r]^{0} & \bullet  \ar@{-}[r]^{11} & \bullet  \ar@{-}[r]^{5} & \bullet  \ar@{-}[r]^{0} & \bullet  \ar@{-}[r]^{3} & \bullet  \ar@{-}[r]^{6} & \bullet \\
& &  \bullet  \ar@{-}[u]_{8} & & & \bullet  \ar@{-}[u]_{8} & & & \bullet  \ar@{-}[u]_{8} & & & \\
& &  \bullet  \ar@{-}[u]_{5} & & & \bullet  \ar@{-}[u]_{5} & & & \bullet  \ar@{-}[u]_{5} & & & \\
}
\]		
		\end{center}
		\caption{Three extended $4$-stars chained together, to form a variable gadget for Theorem~\ref{thm:1-in-3}. The pendant edges drawn on the top will be involved in clauses gadget. \textcolor{black}{We show in the upper drawing how both sides of the regime representing false can be achieved ($2$ and $9$). We show in the lower drawing how it works with $\{5,6\}$}.}
		\label{fig:1-in-3}
		\end{figure}
The clause gadget will be nothing more than a $3$-star, which is formed from a new vertex uniting three (top) pendant edges from their respective variable gadgets. This is drawn in Figure~\ref{fig:1-in-3-clause}. The following is clear.
\begin{lemma}
A clause gadget is in a valid $L(3,2)$-edge-$12$-labelling precisely in the case where its edges are coloured one from $\{5,6\}$ and two from $\{2,3,8,9\}$. 
\label{lem:1-in-3-clause}
\end{lemma}

	\begin{proof}[Proof of Theorem~\ref{thm:1-in-3}.]
We reduce from (monotone) 1-in-3-SAT. Let $\Phi$ be an instance of 1-in-3-SAT involving $n$ occurrences of (not necessarily distinct) variables and $m$ clauses. Let us explain how to build an instance $G$ for {\sc $L(3,2)$-Edge-$12$-Labelling}. Each particular variable may only appear at most $n$ times, so for each variable we take a copy of the variable gadget which is $n$ extended $4$-stars chained together. Each particular instance of the variable belongs to one of the free (top) pendant edges of the variable gadget. For each clause of $\Phi$ we use a $3$-star to unite an instance of these free (top) pendant edges from the corresponding variable gadgets. Thus, we add a single vertex for each clause, but no new edges (they already existed in the variable gadgets). We claim that $\Phi$ is a yes-instance of 1-in-3-SAT if and only if $G$ is a yes-instance of {\sc $L(3,2)$-Edge-$12$-Labelling}.

(Forwards.) Take a satisfying assignment for $\Phi$. Let the range $\{5,6\}$ represent true and the range $\{2,3,8,9\}$ represent false. In particular, every clause has two false and one should be chosen as (e.g.) $2$ and the other $9$. Thus, where a variable is false, some of top pendant edges are labelled $2$ and others $9$ (and this is shown in Figure~\ref{fig:1-in-3}). In each clause, we will have (say) $2,9,5$. Plainly these can be consistently united in a claw by the new vertex that appeared in the clause gadget. We draw the situation in Figure~\ref{fig:1-in-3} to demonstrate that this will not introduce problems at distance $2$. Thus, we can see this is a valid $L(3,2)$-edge-$12$-labelling of $G$.

(Backwards.) From a valid $L(3,2)$-edge-$12$-labelling of $G$, we infer an assignment $\Phi$ by reading, in the variable gadget,  range $\{5,6\}$ as true and the range $\{2,3,8,9\}$ as false. The consistent valuation of each variable follows from Lemma~\ref{lem:1-in-3-variable} and the fact that it is in fact not-all-equal follows from Lemma~\ref{lem:1-in-3-clause}. 
\end{proof}
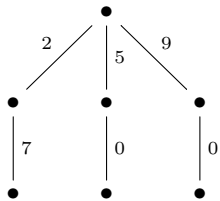
\begin{figure}
$
\xymatrix{
& \bullet  \ar@{-}[dl]_{2} \ar@{-}[d]^{5}  \ar@{-}[dr]^{9}& \\
\bullet \ar@{-}[d]^{7} & \bullet \ar@{-}[d]^{0} & \bullet \ar@{-}[d]^{0}\\
\bullet & \bullet & \bullet \\
}
$

\caption{The clause gadget and its interface with the variable gadgets (where we must consider distance $2$ constraints).}
\label{fig:1-in-3-clause}
\vspace*{-0.5cm}
\end{figure}

\section{Case \boldmath{$\frac 12< \frac{b}{a} < \frac{2}{3}$}}
\label{sec:2-in-4}	
 
	\begin{theorem}
		If $\frac 12<\frac ba < \frac 23$, then the problem  {\sc $L(a,b)$-Edge-$(4b+a+1)$-Labelling} is \NP-complete.
		\label{thm:2-in-4}
	\end{theorem}
This is probably the most involved case in terms of the sophistication of the proof.	We need some lemmas before we can specify our gadgets.
		\begin{lemma}
			If $0 < b < a$ and $\lambda < 3a+b$, with $k=\lambda + 1$, any edge $k$-labelling of the extended $4$-star must involve inner edge labels of $(0\leq) p<q<r<s(<k)$ so that both $q-p\geq 2b$ and $s-r\geq 2b$.
		\end{lemma}
		\begin{proof}
			The assumption $\lambda < 3a+b$ forces: $\lambda-s<b$, $p<b$, $q-p,r-q,s-r<a+b$. Consider colouring the edge beside that edge which is coloured by $r$. This can't be coloured by anything other than something between $p$ and $q$, forcing $q-p\geq 2b$. Similarly, consider colouring the edge beside that edge which is coloured by $q$. This can't be coloured by anything other than something between $r$ and $s$, forcing $s-r\geq 2b$. \qed
		\end{proof}
		\begin{corollary}
			Let $a\leq 2b$. The minimal $k$ so that the extended $4$-star gadget can be edge $k$-labelled is $4b+a+1$.
		\end{corollary}
		\begin{proof}
		We know it is at least $4b+a+1$ from the previous lemma. Further, the colouring alluded to in the previous proof extends to a valid colouring. Set labels $(p,q,r,s)$ to $(0,2b,2b+a,4b+a)$. Then, the edges next to $p$ and $q$ can be coloured $3b+a$, and the edges next to $r$ and $s$ can be coloured $b$. 
		\end{proof}
		\begin{lemma}
			Let $\frac 12<\frac ba< \frac 23$ and $k=4b+a+1$. The extended $4$-star gadget can be edge-$k$-labelled only such that two pendant edges are $b$ and the other two are $3b+a$.
			\label{lem:main-2-in-4}
		\end{lemma}
		\begin{proof}
			The inequality $\frac 12 \textcolor{black}{<} \frac ba$ proves it is a correct labelling.
			
			We have $\lambda=4b+a<3a+b$ so from the previous lemma we deduce the inner edge labels are $0,2b,2b+a,4b+a$. Adjacent to $2b+a$ must be $b$ and the same is true for $4b+a$. Adjacent to $2b$ must be $3b+a$ and the same is true of $0$. 
		\end{proof}
		
		Note that colours are in the set $\{0,\ldots,4b+a\}$. Below, in Figure~\ref{fig:2-in-4-var}, we give two gadgets for the variables, the \emph{end gadget} and the (basic part of the) \emph{variable gadget}. The variable gadget admits a number of edge-$(4b+a+1)$-labellings, but we want the only possibilities to be that drawn and one that swaps $3b+a$ and $b$. This we enforce by attaching an end gadget at the end (e.g. the left-hand end). For example, one may join it by adding new edges (in the present colouring of the end gadget, that would force the other colouring of the variable gadget). That is, we join the end gadget using the two edges drawn at the bottom below to the (basic part of the) variable gadget using the two edges drawn (say) to the left below. The join is accomplished by adding two new edges, one for each position. That is, one edge joins left and bottom, while the other edge joins right and top. In the variable gadget, the variables will extend from the $10$-cycles, but this is possible only on one side.
		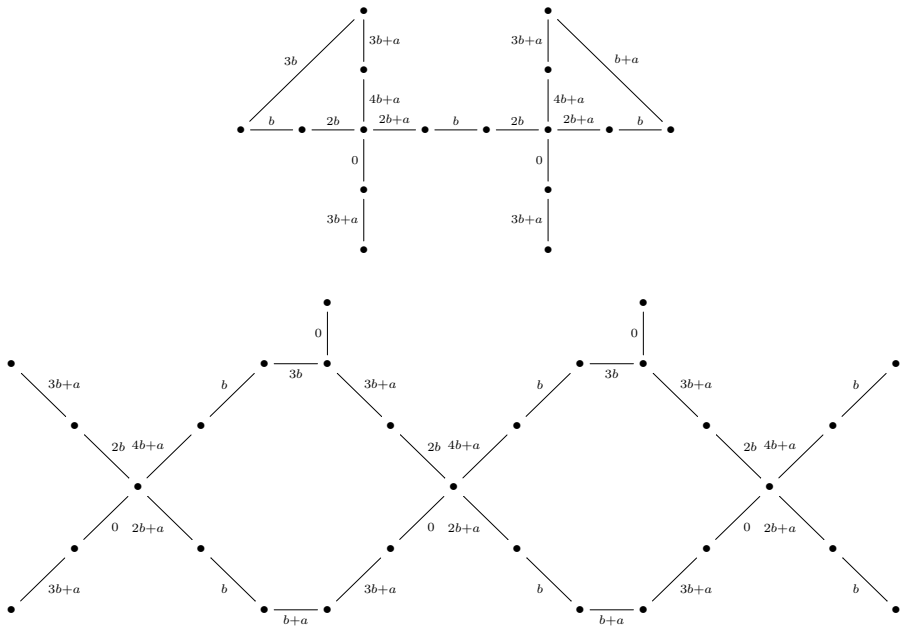
\begin{figure}[h]
\vspace*{-0.9cm}	
		\begin{center}
		\[
		\resizebox{6cm}{!}{
			\xymatrix{
				& & \bullet \ar@{-}[d]^{3b+a}  \ar@{-}[ddll]_{\textcolor{black}{3b}}  & & & \bullet \ar@{-}[d]_{3b+a} \ar@{-}[ddrr]^{\textcolor{black}{b+a}}& &  \\
				& & \bullet \ar@{-}[d]^{4b+a}  & & & \bullet \ar@{-}[d]^{4b+a} & & \\
				\bullet \ar@{-}[r]^{b} & \bullet \ar@{-}[r]^{2b} & \bullet \ar@{-}[r]^{2b+a} & \bullet \ar@{-}[r]^{b} & \bullet \ar@{-}[r]^{2b} & \bullet \ar@{-}[r]^{2b+a} & \bullet \ar@{-}[r]^{b} & \bullet \\ 
				& & \bullet \ar@{-}[u]^{0}  & & & \bullet \ar@{-}[u]^{0} & & \\
				& & \bullet \ar@{-}[u]^{3b+a}  & & & \bullet \ar@{-}[u]^{3b+a} & & \\
			}
		}
		\]
\[
\resizebox{12cm}{!}{\textcolor{black}{
			\xymatrix{
				& & & &  &  \bullet \ar@{-}[d]_{0} & & & & &  \bullet \ar@{-}[d]_{0} & & & & \\
				\bullet \ar@{-}[rd]^{3b+a} & & & & \bullet \ar@{-}[ld]_{b}  \ar@{-}[r]_{\textcolor{black}{3b}} & \bullet \ar@{-}[rd]^{3b+a} & & & & \bullet \ar@{-}[ld]_{b} \ar@{-}[r]_{\textcolor{black}{3b}} & \bullet \ar@{-}[rd]^{3b+a} & & & & \bullet \ar@{-}[ld]_{b}  \\ 
				& \bullet \ar@{-}[rd]^{2b} & &  \bullet \ar@{-}[ld]_{4b+a} &  & & \bullet \ar@{-}[rd]^{2b} & &  \bullet \ar@{-}[ld]_{4b+a} & & & \bullet \ar@{-}[rd]^{2b} & & \bullet \ar@{-}[ld]_{4b+a}  \\
				& & \bullet & & & & &  \bullet & & & & & \bullet & & \\
				& \bullet \ar@{-}[ru]_{0} & &  \bullet \ar@{-}[lu]^{2b+a} &  & & \bullet \ar@{-}[ru]_{0}  & & \bullet \ar@{-}[lu]^{2b+a}  & & & \bullet \ar@{-}[ru]_{0} & & \bullet \ar@{-}[lu]^{2b+a} \\
				\bullet \ar@{-}[ru]_{3b+a} & & & & \bullet \ar@{-}[lu]^{b} \ar@{-}[r]_{b+a} & \bullet \ar@{-}[ru]_{3b+a} & & & & \bullet \ar@{-}[lu]^{b} \ar@{-}[r]_{b+a} & \bullet \ar@{-}[ru]_{3b+a} & & & & \bullet \ar@{-}[lu]^{b}  \\ 
			}
}
}
\]
		\end{center}
		\caption{End gadget (above) and basic part of variable gadget (below).}
		\label{fig:2-in-4-var}
\vspace*{-0.1cm}	
		\end{figure}
		We now meet, in Figure~\ref{fig:2-in-4-var2}, a full variable gadget drawn with a variable protrusion, in this case \textcolor{black}{built from two $0$ edges (the symmetric form gives two $4b+a$ edges)}.
		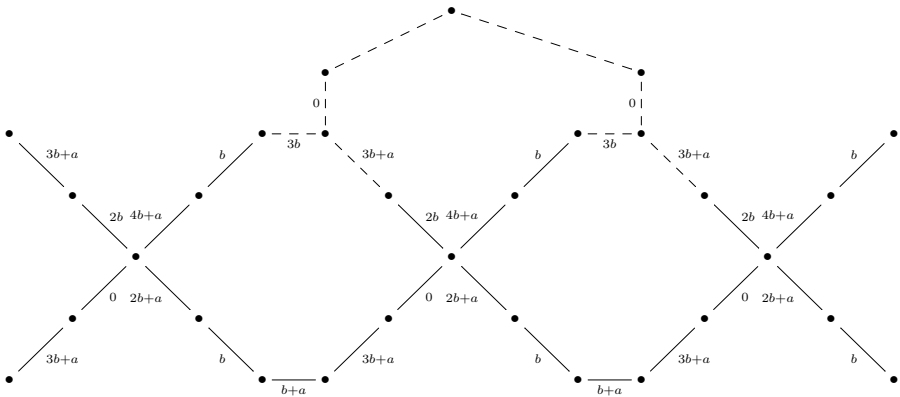
\begin{figure}
		\begin{center}
\[
\resizebox{12cm}{!}{\textcolor{black}{
			\xymatrix{
& & & & & & & \bullet \ar@{--}[dll]_{} \ar@{--}[drrr]_{} \\
				& & & &  &  \bullet \ar@{--}[d]_{0}
				& & & &  &  \bullet \ar@{--}[d]_{0} & & & \\ 
				\bullet \ar@{-}[rd]^{3b+a} & & & & \bullet \ar@{-}[ld]_{b}  \ar@{--}[r]_{\textcolor{black}{3b}} & \bullet \ar@{--}[rd]^{3b+a} & & & & \bullet \ar@{-}[ld]_{b} \ar@{--}[r]_{\textcolor{black}{3b}} & \bullet \ar@{--}[rd]^{3b+a} & & & & \bullet \ar@{-}[ld]_{b}  \\ 
				& \bullet \ar@{-}[rd]^{2b} & &  \bullet \ar@{-}[ld]_{4b+a} &  & & \bullet \ar@{-}[rd]^{2b} & &  \bullet \ar@{-}[ld]_{4b+a} & & & \bullet \ar@{-}[rd]^{2b} & & \bullet \ar@{-}[ld]_{4b+a}  \\
				& & \bullet & & & & &  \bullet & & & & & \bullet & & \\
				& \bullet \ar@{-}[ru]_{0} & &  \bullet \ar@{-}[lu]^{2b+a} &  & & \bullet \ar@{-}[ru]_{0}  & & \bullet \ar@{-}[lu]^{2b+a}  & & & \bullet \ar@{-}[ru]_{0} & & \bullet \ar@{-}[lu]^{2b+a} \\
				\bullet \ar@{-}[ru]_{3b+a} & & & & \bullet \ar@{-}[lu]^{b} \ar@{-}[r]_{b+a} & \bullet \ar@{-}[ru]_{3b+a} & & & & \bullet \ar@{-}[lu]^{b} \ar@{-}[r]_{b+a} & \bullet \ar@{-}[ru]_{3b+a} & & & & \bullet \ar@{-}[lu]^{b}  \\ 
			}
}
}
\]
		\end{center}
		\caption{\textcolor{black}{A full variable gadget drawn with a variable protrusion. Note that each variable protrusion, as the gadget repeats, must be of the same kind. This is demonstrated in Figure~\ref{fig:new-gaetan} where it is shown that the alternative colouring is impossible. The dashed lines in the present drawing also appear in our depiction of the clause gadget in Figures~\ref{fig:2-in-4-clause} and \ref{fig:2-in-4-clause-bis}.}}
		\label{fig:2-in-4-var2}
		\end{figure}

\begin{figure}
\[
		\resizebox{10cm}{!}{\textcolor{black}{
			\xymatrix{
& &				& &  &  \bullet \ar@{-}[d]_{\mathrm{Impossible}} & & & \\
\bullet \ar@{-}[rd]^{3b+a}	 & &				& & \bullet \ar@{-}[ld]_{b}  \ar@{-}[r]_{b+a} & \bullet \ar@{-}[rd]^{3b+a} & &\\ 
& \bullet \ar@{-}[rd]^{0}	 &			 &  \bullet \ar@{-}[ld]_{2b+a} &  & & \bullet \ar@{-}[rd]^{0} &    \\
& &			 \bullet & & & & &  \bullet \\
& \bullet \ar@{-}[ru]_{2b}	&			 &  \bullet \ar@{-}[lu]^{4b+a} &  & & \bullet \ar@{-}[ru]_{2b}  & \\
\bullet \ar@{-}[ru]_{3b+a}	 & & 				& & \bullet \ar@{-}[lu]^{b} \ar@{-}[r]_{3b} & \bullet \ar@{-}[ru]_{3b+a} & & \\ 
			}
		}
}
		\]
		\caption{\textcolor{black}{Demonstration that the variable protrusions are determined once the left-hand leaves of the first extended $4$-star are chosen (remember they are ultimately made equal by the end gadget).}}
		\label{fig:new-gaetan}
\end{figure}
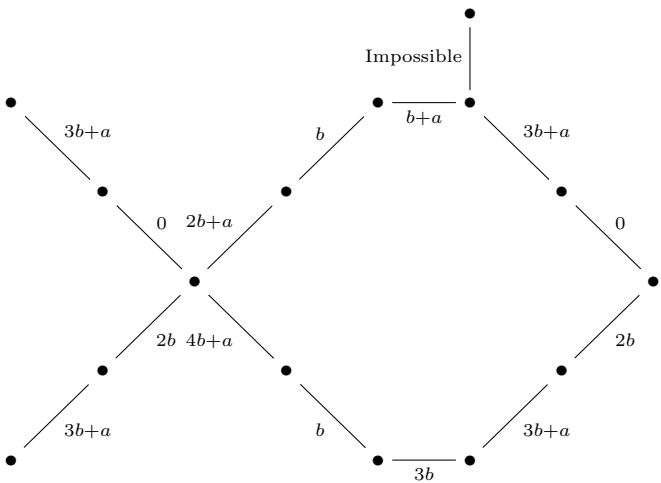

Summing up, we derive the following lemma.
\begin{lemma}
In a full variable gadget complete with an end gadget, any valid 
edge-$(4b+a+1)$-labelling has the property that the pendant edges from the \textcolor{black}{basic part of the variable gadgets, which form the vertical edges in the full variable protrusion, are either all $0$ or are all $4b+a$}.
\label{lem:summing-up-2-in-4}
\end{lemma}
The clause gadget \textcolor{black}{is derived from} an extended $4$-star, whose properties we \textcolor{black}{gave} already in Lemma~\ref{lem:main-2-in-4}. \textcolor{black}{Specifically, we extend the paths in the extended $4$-star from length two to length four where they join the top node from a variable protrusion. Let us call this a triply extended $4$-star. This is drawn in Figure~\ref{fig:2-in-4-clause}, where we also show the interface with the variable gadgets, together with a valid colouring.}

\begin{proof}[Proof of Theorem~\ref{thm:2-in-4}.]	
We reduce from (monotone) 2-in-4-SAT. Let $\Phi$ be an instance of 2-in-4-SAT involving $n$ occurrences of (not necessarily distinct) variables and $m$ clauses. Let us explain how to build an instance $G$ for {\sc $L(a,b)$-Edge-$(4b+a+1)$-Labelling}. Each particular variable only appears at most $n$ times, so for each variable we take a full variable gadget with $n$ \textcolor{black}{variable} protrusions. Each particular instance of the variable belongs to \textcolor{black}{the top vertex of a variable protrusion} (one of these is drawn in Figure~\ref{fig:2-in-4-var2}, but none appears in Figure~\ref{fig:2-in-4-var}). For each clause of $\Phi$ we use a \textcolor{black}{triply} extended $4$-star to unite some instance of these \textcolor{black}{top vertices of the variable protrusions} from the corresponding full variable gadgets. We claim that $\Phi$ is a yes-instance of 2-in-4-SAT if and only if $G$ is a yes-instance of {\sc $L(a,b)$-Edge-$(4b+a+1)$-Labelling}.

(Forwards.) Take a satisfying assignment for $\Phi$. Let \textcolor{black}{$0$} represent true and \textcolor{black}{$4b+a$} represent false. 
Then, every clause has two true and two false \textcolor{black}{variables and} these can be consistently united in an triply extended 4-star as in Figure~\ref{fig:2-in-4-clause}. 
This is a valid $L(a,b)$-edge-$(4b+a+1)$-labelling of $G$.

(Backwards.) From a valid $L(a,b)$-edge-$(4b+a+1)$-labelling of $G$, we infer an assignment $\Phi$ by reading, in the full variable gadget, \textcolor{black}{$0$} as true and \textcolor{black}{$4b+a$} as false. The consistent valuation of each variable follows from Lemma~\ref{lem:summing-up-2-in-4} and the fact that it is 2-in-4 follows from Lemma~\ref{lem:main-2-in-4}, \textcolor{black}{bearing in mind the impossibility of colouring a path in the clause gadget as in Lemma~\ref{lem:2-in-4-clause-bis} and Figure~\ref{fig:2-in-4-clause-bis}}.
\end{proof}
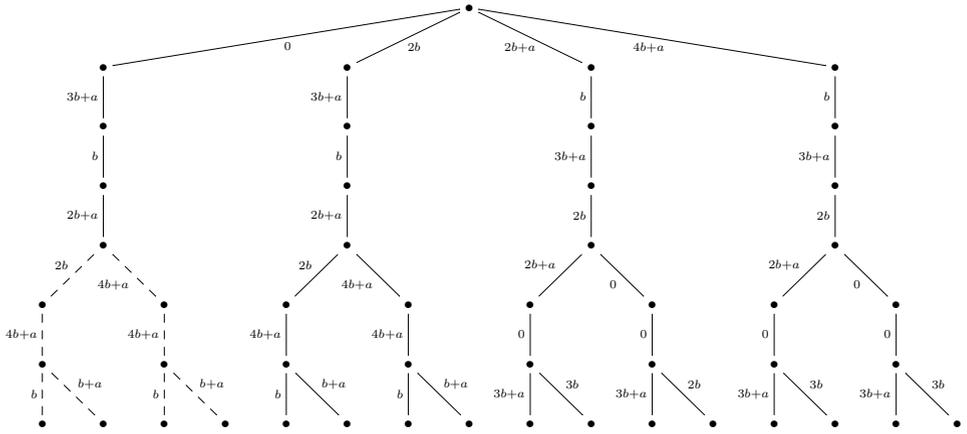
\begin{figure}
\[
\resizebox{13.5cm}{!}{\textcolor{black}{
\xymatrix{
& & & & & & & \bullet \ar@{-}[dllllll]^{0} \ar@{-}[dll]^{2b} \ar@{-}[drr]_{2b+a} \ar@{-}[drrrrrr]_{4b+a}\\
& \bullet \ar@{-}[d]_{3b+a} &  & & & \bullet \ar@{-}[d]_{3b+a}  & &  & & \bullet \ar@{-}[d]_{b}   & & & & \bullet \ar@{-}[d]_{b} & & \\
& \bullet \ar@{-}[d]_{b} &  & & & \bullet \ar@{-}[d]_{b}  & &  & & \bullet \ar@{-}[d]_{3b+a}   & & & & \bullet \ar@{-}[d]_{3b+a} & & \\
& \bullet \ar@{-}[d]_{2b+a} &  & & & \bullet \ar@{-}[d]_{2b+a}  & &  & & \bullet \ar@{-}[d]_{2b}   & & & & \bullet \ar@{-}[d]_{2b} & & \\
& \bullet \ar@{--}[dl]_{2b} \ar@{--}[dr]_{4b+a} &  & & & \bullet \ar@{-}[dl]_{2b} \ar@{-}[dr]_{4b+a}  & & & & \bullet \ar@{-}[dl]_{2b+a} \ar@{-}[dr]_{0}  & & & & \bullet \ar@{-}[dl]_{2b+a} \ar@{-}[dr]_{0}  & & \\
\bullet \ar@{--}[d]_{4b+a}  & & \bullet \ar@{--}[d]_{4b+a} & & \bullet \ar@{-}[d]_{4b+a}  & & \bullet \ar@{-}[d]_{4b+a}  & & \bullet \ar@{-}[d]_{0} & & \bullet \ar@{-}[d]_{0}  & & \bullet \ar@{-}[d]_{0} & & \bullet \ar@{-}[d]_{0}  & & \\
\bullet \ar@{--}[d]_b \ar@{--}[dr]^{b+a} & & \bullet \ar@{--}[d]_b \ar@{--}[dr]^{b+a} & & \bullet \ar@{-}[d]_b \ar@{-}[dr]^{b+a} & & \bullet \ar@{-}[d]_b \ar@{-}[dr]^{b+a} & & \bullet \ar@{-}[d]_{3b+a} \ar@{-}[dr]^{3b} & & \bullet \ar@{-}[d]_{3b+a} \ar@{-}[dr]^{2b} & & \bullet \ar@{-}[d]_{3b+a} \ar@{-}[dr]^{3b} & & \bullet \ar@{-}[d]_{3b+a} \ar@{-}[dr]^{3b} & & \\
\bullet & \bullet & \bullet & \bullet & \bullet & \bullet & \bullet & \bullet & \bullet & \bullet & \bullet & \bullet & \bullet & \bullet & \bullet & \bullet\\
}
}
}
\]
\caption{\textcolor{black}{The clause gadget and its interface with the variable gadgets (where we must consider distance $2$ constraints). In the first of the four variables, on the left-hand branch, we show in dashed lines the corresponding variable protrusion.}}
\label{fig:2-in-4-clause}
\end{figure}			

\begin{figure}
	\[
\textcolor{black}{
		\xymatrix{
			& &  \bullet \ar@{-}[dl]^{} \\
			& \bullet \ar@{-}[d]_{b} &  & & &  \\
			& \bullet \ar@{-}[d]_{x \geq b+a} &  & & &  \\
			& \bullet \ar@{-}[d]_{y\in\{2a,\ldots,4b\}} &  & & &  \\
			& \bullet \ar@{-}[dl]_{z\in\{a,\ldots,2b\}} \ar@{-}[dr]^{4b+a} &  & & &  \\
			\bullet \ar@{-}[d]_{0}  & & \bullet \ar@{-}[d]_{0} \\
			\bullet \ar@{-}[d]_{3b+a} \ar@{-}[dr]^{3b} & & \bullet \ar@{-}[d]_{3b+a} \ar@{-}[dr]^{3b} & \\
			\bullet & \bullet & \bullet & \bullet & \\
		}
		}
		\]
\caption{\textcolor{black}{An impossible colouring on a path in a clause gadget that shows (together with the valid colouring of Figure~\ref{fig:2-in-4-clause}) that the clause gadget enforces 2-in-4-SAT.}}
\label{fig:2-in-4-clause-bis}
\end{figure}
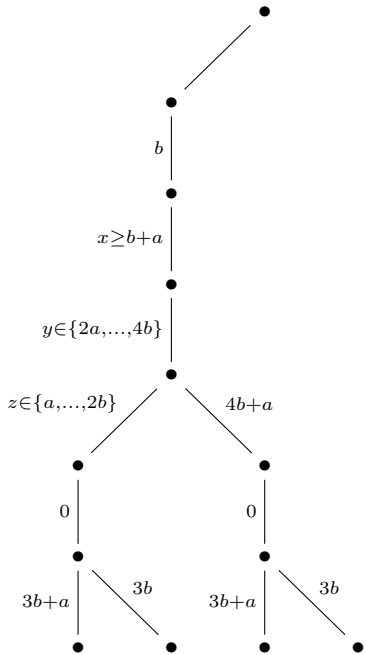			
{\color{black}
\begin{lemma}
The colouring depicted in Figure~\ref{fig:2-in-4-clause-bis} cannot be completed from the initial colouring of the second to top edge as $b$ and the lower six edges as $0$ (above) and $3b+a$ and $3b$ (below).
\label{lem:2-in-4-clause-bis}
\end{lemma}
\begin{proof}
    Case 1 : $x\leq y$. \\
    As $x$ and $y$ are neighbours, we have $x\leq y-a\leq 4b-a$. So the distance between $x$ and $z$ is $\leq 4b-a-a=4b-2a<b$.
    
    This is not possible as we need the distance between $x$ and $z$ to be $\geq b$.\\
    
    Case 2 : $y\leq x$. \\
    As $x$ and $y$ are neighbours, we have $x\geq y+a\geq 3a$. So $4b+a-x\leq 4b-2a=4b-2a<b$.
    
    This is not possible as we need $4b+a-x \geq b$.
\end{proof}
}

\section{Case \boldmath{$0\textcolor{black}{<} \frac{b}{a} \leq \frac{1}{2}$}}\label{a-last}

We follow the exposition of \cite{KM18}, which addresses the $L(2,1)$-edge-$7$-labelling problem. With permission we have used (an adaptation) of their diagrams. Note that in \cite{KM18}, they would call the problem we address the $L(a,b)$-edge-\textbf{3a}-labelling problem as, in their exposition $3a$ refers to the set $\{0,\ldots,3a\}$.

\begin{theorem}
The $L(a,b)$-edge-$(3a+1)$-labelling problem, for $a \geq  2b$, is \NP-hard.
\label{thm:tomas}
\end{theorem}
\begin{proof}
By reduction from (monotone) NAE-3-SAT using the gadgets and properties detailed in Lemmas~\ref{lem:tomas} and \ref{lem:tomas2}, below. 
\end{proof}
For $1\leq k \leq 3$, we define the sets $\circlenum{k}=\llbracket (k-1)a + b, ka - b \rrbracket$. The edges of a $4$-star have to be coloured $0,a,2a,3a$, in any valid $L(a,b)$-edge-$(3a+1)$-labelling. Then any neighbouring edge of the star has to be in one these sets, of the form $\circlenum{k}$. These properties we will now use without further comment.

A variable is represented by the \emph{variable gadget} of Figure~\ref{fig:KM-variable}.
\begin{figure}
\imgw{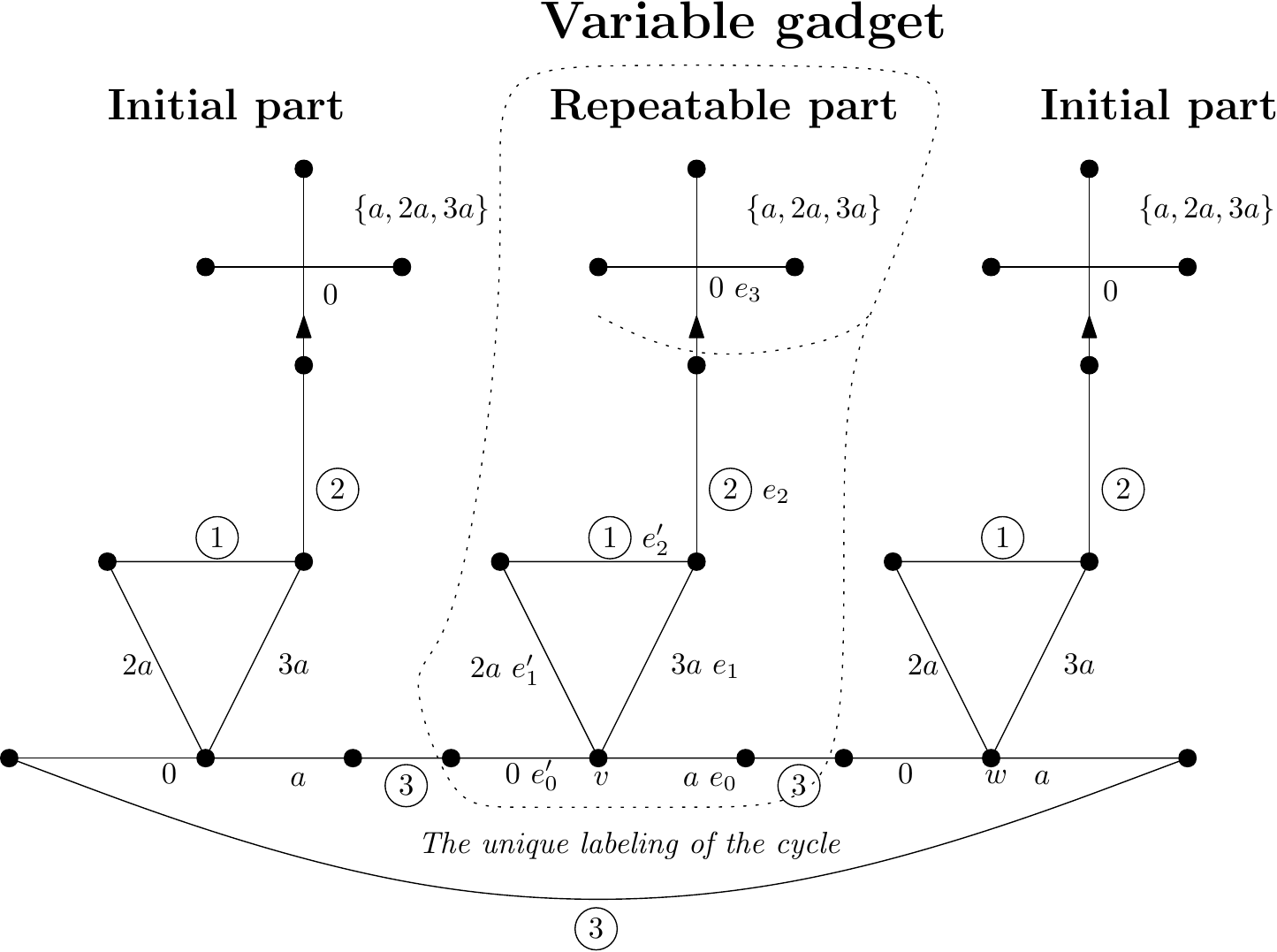}{1}\imgw{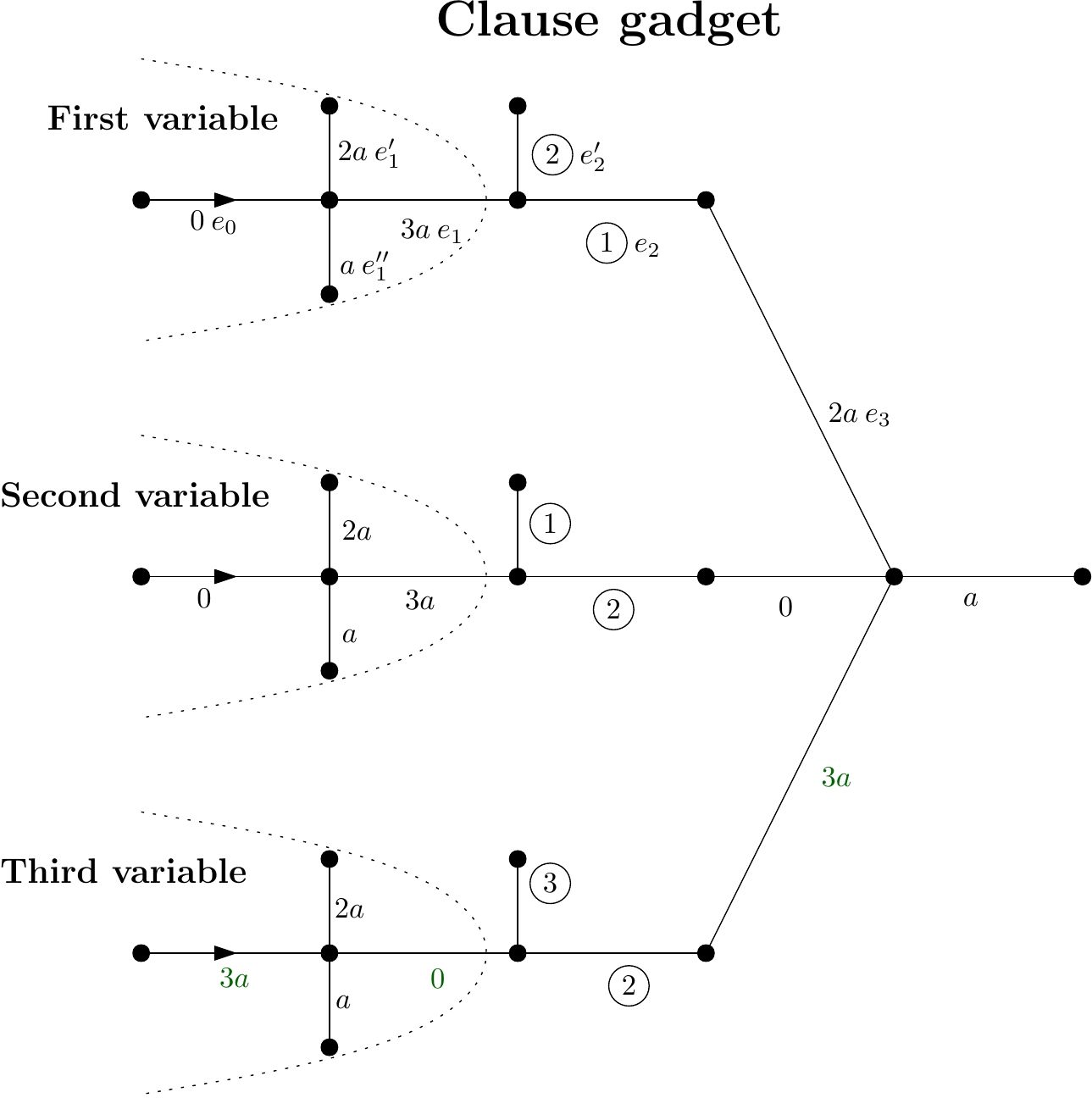}{0.8}
\caption{Variable gadget (adapted from \cite{KM18}).}
\label{fig:KM-variable}
\end{figure}

\begin{lemma}
In any valid $L(a,b)$-edge-$(3a+1)$-labelling of the variable gadget, the three edges free in the top of a $4$-star at the top of a repeatable section must be coloured (in all repeatable and initial parts) by either $\{a,2a,3a\}$ or $\{0,a,2a\}$.
\label{lem:tomas}
\end{lemma}
\begin{proof}
Let us consider various possibilities for the colouring of $\{e'_0,e_0\}$ and $\{e'_1,e_1\}$ (up to the order inverting map that takes $(0,\ldots,3a)$ to $(3a,\ldots,0)$). These are drawn in Table~\ref{tab:KM-variable} (essentially reproduced from \cite{KM18}) together (in some cases) with why they lead to contradiction. The Cases III-VI are straightforward.
\begin{table}[h]
	\begin{center}
		\begin{tabular}{|c||c|c|c|c|}\hline
			&	$e_0'$&$e_1'$&$e_2'$ & \\
		 &	$e_0$&$e_1$&$e_2$ &$e_3$ \\
			\hline\hline
		I.	&	$a$&$0$&in $\circlenum{2}$ & \\
		   &	$2a$&$3a$&in $\circlenum{1}$ &$a,2a$ impossible in the cycle\\\hline
			II.&	$a$&$3a$&in $\circlenum{2}$ & \\
			   &	$2a$&$0$&in $\circlenum{3}$ & $a,2a$ impossible in the cycle\\\hline
		III.&	$0$&$a$&impossible& \\
			&	$3a$&$2a$&---&--- \\\hline
			IV.&$0$&$a$&impossible& \\
			   &$2a$&$3a$&---&--- \\\hline
			V.&	$0$&$3a$&in $\circlenum{1}$ & \\
		   &	$a$&$2a$&impossible&--- \\\hline
			VI.&	$0$&$2a$&in $\circlenum{1}$ & \\
			&	$a$&$3a$&in $\circlenum{2}$ &$0$ \\\hline
		\end{tabular}
			\caption{Variable gadget table (adapted from \cite{KM18}).}
			\label{tab:KM-variable}
	\end{center}

\end{table}
For Cases I and II we need to argue why $\{e'_0,e_0\}$ \textcolor{black}{cannot} be $\{a,2a\}$. In this case, the cycle must continue (bearing in mind that every \textcolor{black}{edge with a vertex of degree $4$} must be from $\{0,a,2a,3a\}$) in a certain way. To the right it must continue: $(a,2a,\circlenum{1},3a,2a,\circlenum{1},3a,2a,\ldots)$. However, to the left it must continue: $(2a,a,\circlenum{3},0,a,\circlenum{3},0,a,\ldots)$. These paths can now never join together in a cycle. This rules out Cases I and II.

The remaining labellings, Case VI and its various symmetries, are possible and result in the claimed behaviour. 
\end{proof}

The clause is represented by the \emph{clause gadget} of Figure~\ref{fig:KM-clause}.

\begin{lemma}
Consider any valid $L(a,b)$-edge-$(3a+1)$-labelling of the clause gadget, such that the input parts of the variable gadgets satisfy the previous lemma. Two of the input variable gadget parts must come from one of the regimes $\{a,2a,3a\}$ or $\{0,a,2a\}$, and the other input part from the other regime. In particular, if all three input variable gadget parts come from only one of the regimes, then this can not be extended to a valid $L(a,b)$-edge-$(3a+1)$-labelling.
\label{lem:tomas2}
\end{lemma}
\begin{proof}
Let us consider various possibilities for the colouring of $e_1$ and $\{e'_1,e''_1\}$ (up to the order inverting map that takes $(0,\ldots,3a)$ to $(3a,\ldots,0)$). These are drawn in Table~\ref{tab:KM-clause} (essentially reproduced from \cite{KM18}) together (in some cases) with why they lead to contradiction.

Cases III and IV show the valid possibilities. The three edges where the $4$-star unites the variable gadget repeatable parts has only two possibilities for each of the variable regimes $\{a,2a,3a\}$ or $\{0,a,2a\}$ (namely, $\{0,2a\}$ and $\{a,3a\}$, respectively). The claimed behaviour is clear. 
\end{proof}

\newpage
\begin{figure}
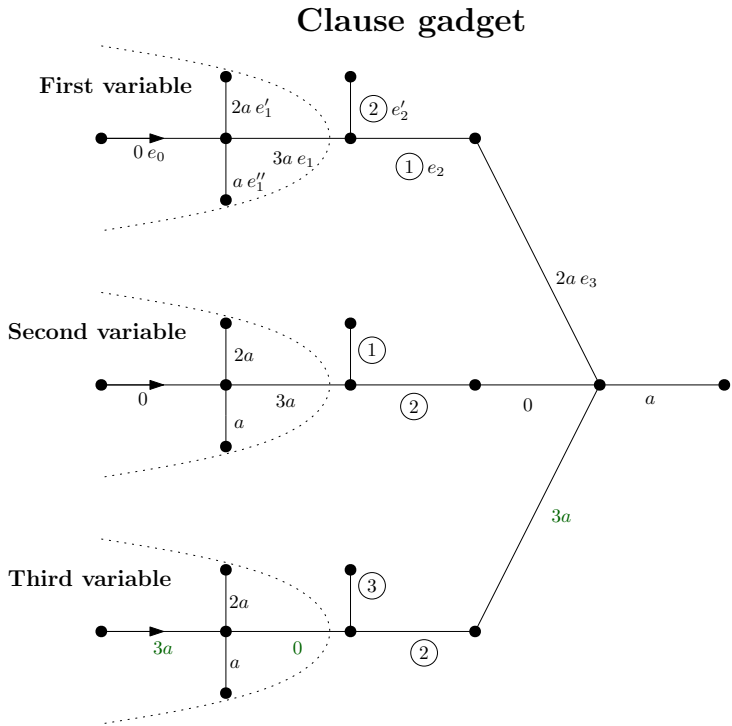

\imgw{eng-klauzule_3a_a_b-eps-converted-to}{0.8}
\caption{Clause gadget (adapted from \cite{KM18}).}
\label{fig:KM-clause}
\end{figure}
\begin{table}[h]
	\begin{center}
		\begin{tabular}{|c||c|c|c|c|}\hline
			&	\multirow{4}{*}{$e_0$}&$e_1''$& & \\
			&												 &$e_1'$& & \\
			&				&\multirow{2}{*}{$e_1$}&$e_2'$& \\
		 & &&$e_2$&$e_3$\\\hline 
			\hline
			I.&	\multirow{4}{*}{$0$}&$2a$ or $3a$& & \\
			  &											&$3a$ or $2a$& & \\
			  &			 &\multirow{2}{*}{$a$}&Both in $\circlenum{3}$. Impossible& \\
			  & && Both in $\circlenum{3}$. Impossible &---\\\hline 
			II.&	\multirow{4}{*}{$0$}&$a$ or $3a$& & \\
				  &										&$3a$ or $2a$& & \\
				  &		 &\multirow{2}{*}{$2a$}&Both in $\circlenum{1}$. Impossible& \\
			   & &&Both in \circlenum{1}. Impossible&---\\\hline 
			III.&	\multirow{4}{*}{$0$}&$2a$ or $a$& & \\
				   &										&$a$ or $2a$& & \\
					  &	 &\multirow{2}{*}{$3a$}&In $\circlenum{1}$& \\
				& &&In $\circlenum{2}$&$0$ or $\circlenum{3}$\\\hline 
			IV.&	\multirow{4}{*}{$0$}&$2a$ or $a$& & \\
			&										&$a$ or $2a$& & \\
			&	 &\multirow{2}{*}{$3a$}&In $\circlenum{2}$& \\
			& &&In $\circlenum{1}$&$2a$ or $\circlenum{3}$\\\hline
		\end{tabular}
	\end{center}
	\caption{Clause gadget table (adapted from \cite{KM18}).}
	\label{tab:KM-clause}
	\end{table}

\section{Final Remarks}

We give several directions for future work.
First, determining the boundary for $k$ between P and \NP-complete, in {\sc $L(p,q)$-Edge-$k$-Labelling}, for all $p,q$ is still open except if $(p,q)=(1,1)$ and $(p,q)=(2,1)$.
For $(p,q)=(1,1)$ it is known to be $4$ (it is in P for $k<4$ and is \NP-complete for $k\geq 4$) \cite{Ma02}; and for $(p,q)=(2,1)$ it is known to be $6$ (it is in P for $k<6$ and is \NP-complete for $k\geq 6$) \cite{KM18}.

A second open line of research concerns 
{\sc $L(p,q)$-Labelling} for classes of graphs that omit a single graph $H$ as an induced subgraph (such graphs are called {\it $H$-free}). A rich line of work in this vein includes \cite{BJMPS20}, where it is noted, for $k\geq 4$, that {\sc $L(1,1)$-$k$-Labelling} is in P over $H$-free graphs, when $H$ is a linear forest; for all other $H$ the problem remains \NP-complete. If $k$ is part of the input and $p=q=1$, the only remaining case is $H=P_1+P_4$~\cite{BJMPS21}.
Corollary~\ref{c-line} covers, for every $(p,q)\neq (0,0)$, the case where $H$ contains an induced claw (as every line graph is claw-free). 
For bipartite graphs, and thus for $H$-free graphs for all $H$ with an odd cycle, the result for {\sc $L(p,q)$-$k$-Labelling} is known from \cite{JKM09}, at least in the case $p>q$. 

As our final open problem, for $d\geq 1$, the complexity of {\sc $L(p,q)$-Labelling} on graphs of diameter at most~$d$ has, so far, only been determined for $a,b\in \{1,2\}$~\cite{BGMPS21}.


\end{document}